\documentclass[11pt,letterpaper]{article}
\usepackage{epsfig,rotating,setspace,latexsym,amsmath,epsf,amssymb,bm}
\usepackage{cite,graphicx,authblk,color,subfigure}
\usepackage[title]{appendix}

\title{Near Optimal Coded Data Shuffling for Distributed Learning\vspace{10pt}
\author{Mohamed Adel Attia \qquad Ravi Tandon}
\affil{Department of Electrical and Computer Engineering\\
University of Arizona, Tucson, AZ, USA.\\
E-mail: {\{madel, tandonr\}}@email.arizona.edu}}

\newcommand{\A}{{\mathcal{A}}}
\newcommand{\Z}{{\mathbf{Z}}}

\newtheorem{definition}{Definition}
\newtheorem{theorem}{Theorem}
\newtheorem{lemma}{Lemma}
\newtheorem{claim}{Claim}
\newtheorem{remark}{Remark}
\newtheorem{example}{Example}

\newtheorem{corollary}[theorem]{Corollary}

\newenvironment{proof}[1]{\medskip\par\noindent
{\bf Proof:\,}\,#1}{{\mbox{\,$\blacksquare$}\par}}
\begin{document}
\maketitle
\newcommand\blfootnote[1]{%
  \begingroup
  \renewcommand\thefootnote{}\footnote{#1}%
  \addtocounter{footnote}{-1}%
  \endgroup
}

\blfootnote{This work was supported by the NSF Grant  CAREER-1651492.}

\thispagestyle{empty}

\begin{abstract}
Data shuffling between distributed cluster of nodes is one of the critical steps in implementing large-scale learning algorithms. Randomly shuffling the data-set among a cluster of workers allows different nodes to obtain fresh data assignments at each learning epoch. This process has been shown to provide
 improvements in the learning process. However, the statistical benefits of distributed data shuffling come at the cost of extra communication overhead from the master node to worker nodes, and can act as one of the major bottlenecks in the overall time for computation. 
 There has been significant recent interest in devising approaches to minimize this communication overhead. One approach is to provision for extra storage at the computing nodes. The other emerging approach is to leverage coded communication to minimize the overall communication overhead.

  The focus of this work is to understand the fundamental trade-off between the amount of storage and the communication overhead for distributed data shuffling. In this work, we first present an information theoretic formulation for the data shuffling problem, accounting for the underlying problem parameters (number of workers, $K$, number of data points, $N$, and the available storage, $S$ per node). We then present an information theoretic lower bound on the communication overhead for data shuffling as a function of these parameters. We next present a novel coded communication scheme and show that the resulting communication overhead of the proposed scheme is within a multiplicative factor of at most $\frac{K}{K-1}$ from the information-theoretic lower bound. Furthermore, we present the aligned coded shuffling scheme for some storage values, which achieves the optimal storage vs communication trade-off for $K<5$, and further reduces the maximum multiplicative gap down to $\frac{K-\frac{1}{3}}{K-1}$, for $K\geq 5$.
\end{abstract}

\newpage

%

\section{Introduction}
\label{sec:Introduction}

Owing to the parallelized nature of the distributed computing, and the abundance of computational resources over a large cluster of workers, distributed computational frameworks can enable data-intensive learning tasks and big data applications in a timely manner. 
Distributed computing comes at the unavoidable communication cost due to data transfer to the distributed machines, and the data shuffling process among the distributed workers, which is a basic building block in machine learning paradigms.
The data shuffling block can arise in many applications such as: a) random shuffling of the data-set across different points before each learning epoch so that each worker is assigned new training data, which is a common practice that provides statistical benefits, e.g., distributed gradient descent algorithm and its stochastic variations \cite{SGD-tricks,RandomReshuffling1,RandomReshuffling2}; b) shuffling the data-set across attributes to assign different features (or attributes) to each worker, e.g., in mobile cloud gaming systems \cite{MobileCloudGaming}; and c) shuffling the data between the mappers and the reducers in the MapReduce framework \cite{MapReduce2004}, where the reducers are interested in collecting the data with the assigned ``key(s)" from the mappers.

Another limiting byproduct of distributing the learning process over a large number of  machines is the latency caused by the stragglers, i.e., the workers slower than the average due to several factors such as resource contention, disk failure, power limits, and heterogeneous processing capabilities \cite{Stragglers1,Stragglers2}.
The straggler problem usually limits the completion time by the slowest worker.
Several approaches to mitigate the stragglers effect include a) scheduling redundant computations in \cite{Stragglers/replicas1,Stragglers/replicas2,Stragglers/replicas2-2,Stragglers/replicas3}, such that any  unexpected tardiness or failure of a worker can be compensated by another worker doing the same computations; b) work stealing where the faster workers once they finish their tasks take over the remaining computations from the slower workers \cite{work-stealing}; and recently c) work exchange based on the work conservation principle, where coarse heterogenity knowledge/estimation can be used to reassign the work load according to the speed of the workers \cite{work-exchange2017}.


A promising research has recently emerged in large scale distributed computing addressing both wired networks, where the computations are done over the cloud \cite{Hadoop2010,Spark2010}, and wireless networks, where the computations are done over small mobile machines removing the burden from the cloud \cite{wireless1,wireless2,benefit/wireless}.
Distributed computing platforms can also be classified according to the underlying network topology. In the master-worker setting, a centralized master node posses the whole data set and assigns different parts of the data to a set of distributed workers, which collaboratively learn a shared prediction model to be averaged out at the master node later; while in the worker-to-worker setting (also referred to as the MapReduce framework \cite{MapReduce2004}), the distributed workers are mapped to train different parts of the data to calculate some functions, then the reducers collects the data with the same ``key" to compute each function separately.

The application of coding theory to overcome the communication and latency bottlenecks in order to speed-up the learning process was first considered in \cite{CodedComp/Shuffle2015}.
In particular, the idea of using coded data shuffling was first proposed in \cite{CodedComp/Shuffle2015}, where excess storage at the workers was utilized to create coded broadcasting opportunities in order to reduce the communication overhead. In the same work, $(n,k)$ Maximum Distance Separable (MDS) codes were proposed for distributed matrix multiplication to mitigate the impact of stragglers.
Coded computation using MDS codes in presence of stragglers
was proposed in \cite{GradientCoding} for synchronous gradient descent, and \cite{Stragglers-MDSproduct,stragglers/wired/MDS1,stragglers/wired/MDS2} for linear computation tasks, e.g., matrix multiplication.
The use of Polynomial codes for high dimensional coded matrix multiplication was proposed in \cite{Stragglers-Polynomial}.
Coded computation over wireless networks was proposed in  \cite{stragglers/wirless}, where only one worker can transmit at a time. 
The use of codes to reduce the communication overhead due to data shuffling was considered in \cite{CodedMapReduce,CDC1,CDC2,CDC/Alternative,CDC/stragglers1,CDC/stragglers2,Shuffling/wired/master1,Shuffling/wired/master2,WDC1,WDC2}. 
In \cite{CodedMapReduce,CDC1,CDC2,CDC/Alternative}, the authors considered the MapReduce setting, where in order to reduce the communication between the mappers and the reducers, coding opportunities are created with more redundant computations at the mappers, leading to a trade-off between communication and computation.
\cite{CDC/stragglers1,CDC/stragglers2} provided a unified coding framework for distributed computing, where the communication load due to shuffling can be alleviated by trading the computational complexity in the presence of straggling servers. 
The information theoretic limits for data shuffling in the wired master-worker setting was considered in \cite{CodedComp/Shuffle2015,Shuffling/wired/master1,Shuffling/wired/master2}. Coded data shuffling in wireless setting was recently considered in \cite{benefit/wireless,WDC1,WDC2} for both centralized and decentralized approaches. 

\noindent \textbf{\underline{Related Works on Data Shuffling and Connections to Index Coding}}:
Using codes for random data shuffling over wired master-worker based distributed computational systems was first considered in \cite{CodedComp/Shuffle2015}. A probabilistic coding scheme was introduced showing how using excess storage can reduce the average communication overhead.
In our initial preliminary work \cite{Shuffling/wired/master1}, the optimal worst-case communication overhead was characterized as a function of the available storage for $K=2,3$ workers using a systematic storage placement, and data delivery schemes. In another work \cite{Shuffling/wired/master2}, the no-excess storage case was considered, where it was shown that even for minimum storage value coding opportunities still exist. A systematic coding scheme was developed for any number of workers, which was proven to be information theoretically optimal in the worst-case scenario.

The data shuffling problem can also be viewed as an index coding problem \cite{IndexCoding}, where the amount of data stored at the workers form the side information, and the new data assignments are the messages needed by each worker. The side information in the data shuffling problem is generally not static, where the storage of the workers can be adapted to reduce the communication overhead in the next shuffle. We propose in this work a structural invariant placement mechanism, where the storage of the workers is updated according to the latest shuffle to maintain the structure.
Furthermore, it was shown in \cite{IndexCoding} that the index coding is a NP-hard problem, and may require in the worst-case a rate of order $O(K)$, where $K$ is the number of workers. A pliable index coding approach for data shuffling was assumed in \cite{pliableIC/shuffling}, where a semi-random shuffles were considered and was shown to achieve a rate of order $O(\log^2(K))$. In this work however, we consider the worst-case rate over all possible shuffles and show that even for the minimum storage (side information at the workers), a rate of order $O(\frac{K-1}{K})$ can be achieved, which does not scale with the number of workers for large values of $K$.

\subsection{Main Contributions of this Paper}

In this paper, we focus on the coded data shuffling for the wired master-worker setting, where coding opportunities are created by exploiting the excess storage at the workers. Before each learning epoch, the data is shuffled at the master node for different training data assignment at each worker, which causes the communication overhead. On one extreme, when all the workers have enough storage to store the whole data set, then no communication is needed for any random shuffle. On the other hand, when the storage is just enough to store the assigned data, which we also refer to as the \textit{no excess storage case}, then the communication is expected to be maximal. Thus, we aim to characterize the fundamental information-theoretic trade-off between the communication overhead due to shuffling and the available storage at the distributed workers.
The contributions of this paper are summarized next:

\noindent$\bullet$ \hspace{5pt} We first derive an information theoretic lower bound on the worst-case communication overhead for the data shuffling problem. We start by obtaining a family of lower bounds on the rate of some chosen shuffles. Since the rate of any shuffle is at most as large as the worst-case shuffle, the obtained lower bounds serve as valid lower bounds for the worst-case rate as well. We then average out all the lower bounds we get using the chosen shuffles. The key step here is choosing the shuffles which lead to the best lower bound on the communication overhead as a function of the storage. In particular, we consider a set of cyclic shuffles where no overlap between the assigned data batch to any worker in the two subsequent shuffles.
Based on a novel bounding methodology similar to the recent results in the caching literature \cite{Optimality/UncodedCache,Optimality/UncodedCache2}, we are able to express the lower bound as a linear program (LP).
We then solve the LP to obtain the best lower bounds on the communication overhead for different regimes of storage.

\noindent$\bullet$ \hspace{5pt} Next, we introduce our achievable scheme based on a placement/update procedure that maintains the structure of the storage, which we refer to as \textit{``the structural invariant placement and update"}. 
The storage placement involves partitioning the data points across dimensions, which allows each worker to store at least some parts of each data point. Through a careful novel storage update, the structure of the storage can be maintained over time.
This allows for applying a data delivery mechanism similar to \cite{FundLimitsCaching2015}, which approaches the optimal worst-case communication-storage trade-off (based on the obtained lower bound) within a vanishing gap ratio of $\frac{K}{K-1}$ as the number of distributed workers $K$ increases.

\noindent$\bullet$ \hspace{5pt} Finally, we introduce new ideas on how to fully characterize the optimal worst-case communication overhead. We show that by considering more sophisticated interference alignment mechanisms, we can force the interference seen by each worker to occupy the minimum possible dimensions.
We refer to this procedure as the \textit{``Aligned Coded Shuffling"} scheme.  This scheme also involves a different structural invariant update mechanism of the storage, which is based one data partitioning and relabeling over time.
Following these ideas, we can close the gap between the obtained bounds for some storage values, which closes the gap for $K<5$, and brings the maximum gap ratio down to $\frac{K-\frac{1}{3}}{K-1}$, for $K\geq 5$.

\section{System Model}
\label{sec:System}

We assume a master node which has access to the entire data-set $\A=\{D_1,D_2,\ldots,D_N\}$ of size $Nd$ bits, i.e., $\A$ is a set containing $N$ data points, denoted by $D_1,D_2,\ldots,D_N$, where $d$ is the dimensionality of each data point. Treating the data points $D_n$ as i.i.d. random variables, we therefore have the entropies of these random variables as
\begin{align}
\label{eq:data-set}
&H(\A)=N\times H(D_n)=Nd, \quad \forall n\in \{1, 2, \ldots, N\}.
\end{align}

At each iteration, indexed by $t$, the master node divides the data-set $\A$ into $K$ data batches given as $\A^{t}(1), \A^{t}(2), \ldots, \A^t(K)$, where $\A^t(k)$ denotes the data partition designated to be processed by worker $w_k$ at time $t$, and these batches correspond to the random permutation of the data-set, $\pi^t:\A\rightarrow(\A^{t}(1), \ldots, \A^t(K))$.  Note that these data batches are disjoint, and span the whole data-set, i.e.,
\begin{subequations}
\label{eq:data-batches}
\begin{align}
&\A^t(i) \cap \A^t(j) = \phi, \quad \forall i\neq j,\\
&\A^t(1) \cup \A^t(2)\cup \ldots \cup \A^t(K) =\A, \quad \forall t.\label{eq:data-partitions}
\end{align}
\end{subequations}
Hence, the entropy of any batch $\A^t(k)$ is given as
\begin{align}
\label{eq:data-batches2}
H(\A^t(k))= \frac{1}{K} H(\A)= \frac{N}{K}d,\quad \forall k\in\{1,\ldots,K\}.
\end{align}

After getting the data batch, each worker locally computes a function (as an example, this function could correspond to the gradient or sub-gradients of the data points assigned to the worker). The local functions from the $K$ workers are processed subsequently at the master node. We assume that each worker $w_k$ has a storage $Z^t_k$ of size $Sd$ bits, for a real number $S$, which is used to store some function of the data-set. Therefore, if we consider $Z^t_k$ as a random variable then,
\begin{align}
\label{eq:cache-content}
H(Z^t_k| \A) = 0, \quad \forall k\in[1:K].
\end{align} 
For processing purposes, the assigned data blocks are needed to be stored by the workers, therefore, each worker $w_k$ must at least store the data block $\A^t(k)$ at time $t$, which gives the storage constraint as
\begin{align}
\label{eq:cache-storage}
H(Z^t_k)=Sd \geq H(\A^t(k)),\qquad \forall k\in\{1,\ldots,K\}.
\end{align}

According to (\ref{eq:data-batches2}) and (\ref{eq:cache-storage}), we get the \textit{minimum storage per worker} $S \geq \frac{N}{K}$.
We also have the \textit{processing constraint} as
\begin{equation}
\label{eq:cache-min-desired}
H(\A^{t}(k)|Z^{t}_k)=0,\quad \forall k\in\{1,\ldots,K\},
\end{equation}
which means $\A^{t}(k)$ is a deterministic function of the storage $Z^{t}_k$.

In the next epoch $t+1$, the data-set is randomly reshuffled at the master node according to a random permutation $\pi^{t+1}: \A\rightarrow(\A^{t+1}(1), \A^{t+1}(2), \ldots, \A^{t+1}(K))$, which also satisfies the properties in \eqref{eq:data-batches}. The main communication bottleneck occurs during  {{\textit{Data Delivery}}} since the master node needs to communicate the new data batches to the workers.
Trivially, if the storage (per worker) exceeds $Nd$ bits, i.e., $S\geq N$, then each worker can store the whole data-set, and no communication has to be done between the master node and the workers for any shuffle. Therefore from the constraint on minimum storage per worker, we can write the possible range for  storage  as $\frac{N}{K}\leq S \leq N$.

We next proceed to describe the data delivery mechanism, and the associated encoding and decoding functions. 
The main process can be divided into two phases, namely the data delivery phase and the storage update phase as described next: 

\subsection{Data Delivery Phase}
At time $t+1$, the master node sends a function of the data batches for the subsequent shuffles $(\pi_t,\pi_{t+1})$, $X_{\pi_t,\pi_{t+1}} = \phi(\A^{t}(1), \ldots, \A^{t}(K),\A^{t+1}(1),\ldots, \A^{t+1}(K))=\phi_{\pi_t,\pi_{t+1}}(\A)$ over the shared link, where $\phi$ is the data delivery encoding function,
\begin{equation}
\phi: \left[2^{\frac{N}{K}d}\right]^{2K} \rightarrow [2^{R_{\pi_t,\pi_{t+1}}d}],
\end{equation}
where $R_{\pi_t,\pi_{t+1}}$ is the rate of the shared link based on the shuffles $(\pi_t,\pi_{t+1})$.
Therefore, we have
 \begin{align}
\label{eq:transmit-load}
 H\left(X_{\pi_t,\pi_{t+1}}|\A\right)=0,\quad H\left(X_{\pi_t,\pi_{t+1}}\right) =R_{\pi_t,\pi_{t+1}}d,
 \end{align}
which means that $X_{\pi_t,\pi_{t+1}}$ is a deterministic function of the whole data-set $\A$.
 
Each worker $w_k$ should reliably decode the desired batch $\A^{t+1}(k)$ out of the transmitted function $X_{\pi_t,\pi_{t+1}}$, as well as the data stored in the previous time slot $Z^{t}_k$, i.e., ${\A}^{t+1}(k) =\psi(X_{\pi_t,\pi_{t+1}}, Z^{t}_k)$, where $\psi$ is the decoding function at the workers,
\begin{equation}
\psi: [2^{R_{\pi_t,\pi_{t+1}}d}]\times [2^{Sd}]\rightarrow [2^{\frac{N}{K}d}].
\end{equation} 
Therefore, for reliable decoding, we have the following \textit{decodability constraint} at each worker:
\begin{equation}
\label{eq:decoding-const}
H\left(\A^{t+1}(k)|Z^{t}_k, X_{\pi_t,\pi_{t+1}}\right)=0, \quad \forall k\in\{1,\ldots,K\}.
\end{equation}

\subsection{Storage Update Phase}
\label{sec:sys-model-2}
At the next iteration $t+1$, every worker updates its stored content according to the placement strategy, where the new storage content for worker $w_k$ is given by  $Z^{t+1}_k$, which is a function of the old storage content $Z^{t}_k$ as well as transmitted function $X_{\pi_t,\pi_{t+1}}$, i.e., ${Z}^{t+1}_k=\mu(X_{\pi_t,\pi_{t+1}},Z^{t}_k)$, where $\mu$ is the update function
\begin{equation}
\mu: [2^{R_{\pi_t,\pi_{t+1}}d}]\times [2^{Sd}]\rightarrow [2^{Sd}],
\end{equation}
Therefore, we have the following \textit{storage-update constraint}:
\begin{equation}
\label{eq:cache-update}
H(Z^{t+1}_k|Z^{t}_k, X_{\pi_t,\pi_{t+1}})=0,\quad \forall k\in\{1,\ldots,K\}.
\end{equation}

The excess storage after storing ${\A}^{t+1}(k)$ in ${Z}^{t+1}_k$, given by $\left(S-\frac{N}{K}\right)d$ bits, can be used to store opportunistically a function of the remaining $K-1$ data batches. 
For the scope of this work, we assume that the placement of the excess storage is uncoded, which means that the excess storage is dedicated to store uncoded functions of the remaining $K-1$ batches.
We give the notation $\A^{t+1}(i,k)$, where $i\neq k$, as the part of data that worker $w_k$ stores about $\A^{t+1}(i)$ in the excess storage at  time $t+1$.
As a result, we can write the content of $Z_k^{t+1}$ for uncoded storage placement as
\begin{align}
\label{eq:cache-content2}
Z_k^{t+1} = \left\{\A^{t+1}(k), \underset{j\in [1:K]\setminus k}{\cup}\A^{t+1}(j,k)\right\}.
\end{align}
Furthermore, we assume a generic placement strategy for the excess storage as follows: the batch $\A^{t+1}(i)$ consists of  $2^{K-1}$ partitions, denoted as $\A^{t+1}_{\mathcal{W}}(i)$, $\mathcal{W}\in 2^{[1:K]\setminus i}$, where $2^{[1:K]\setminus i}$ is the power set of all possible subsets of the set $[1:K]\setminus i$ including the empty set. The worker $w_{k}$, for $k\neq i$, stores the partition $\A^{t+1}_{\mathcal{W}}(i)$ in the excess storage, only if $k \in \mathcal{W}$. Therefore, the sub-batches $\A^{t+1}(i)$, and $\A^{t+1}(i,k)$ can be expressed as
\begin{align}
\label{eq:def_A_k_W}
&\A^{t+1}(i) = \underset{\mathcal{W} \subseteq [1:K]\setminus i}{\cup} \A^{t+1}_{\mathcal{W}}(i),\qquad \A^{t+1}(i,k) = \underset{\mathcal{W} \subseteq [1:K]\setminus i :\: k \in \mathcal{W}}{\cup} \A^{t+1}_{\mathcal{W}}(i).
\end{align}

Let us consider $\A^{t+1}_{\mathcal{W}}(i)$ as a random variable with entropy $H(\A^{t+1}_{\mathcal{W}}(i))= \vert \A^{t+1}_{\mathcal{W}}(i) \vert d$, and size $\vert \A^{t+1}_{\mathcal{W}}(i) \vert$  normalized by the data point size $d$. Therefore, the following two constraints are obtained:

\noindent$\bullet$\hspace{5pt}\textbf{Data size constraint:} The first constraint is related to the total size of the data given by $Nd$ bits,
\begin{align}
N &= \frac{1}{d} H(\A) = \frac{1}{d}\sum_{i=1}^K H(\A^{t+1}(i)) \overset{(a)}{=}\frac{1}{d} \sum_{i=1}^K \sum_{\mathcal{W} \subseteq [1:K]\setminus i} H(\A^{t+1}_{\mathcal{W}}(i))\nonumber \\
&=  \sum_{\ell=1}^K\sum_{i=1}^K \sum_{\mathcal{W} \subseteq [1:K]\setminus i : \:\vert\mathcal{W} \vert = \ell} \vert \A^{t+1}_{\mathcal{W}} (i)\vert = \sum_{\ell =1}^K x_{\ell},\label{eq:size-const}
\end{align}
where $(a)$ follows from \eqref{eq:def_A_k_W}, and $x_{\ell}\geq 0$ is defined as
\begin{align}
 x_{\ell}\: \overset{\Delta}{=}\: \sum_{i=1}^K \sum_{\substack{\mathcal{W}\subseteq [1:K]\setminus i:\: \vert\mathcal{W}\vert =\ell}}\vert \A^t_{\mathcal{W}}(i)\vert, \quad \ell\in[0:K-1].\label{eq:x_ell}
\end{align}

\noindent$\bullet$\hspace{5pt}\textbf{Excess storage size constraint:} The second constraint is related to the total excess storage of all the workers, which cannot exceed $K\left(S-\frac{N}{K}\right)d$ bits,
\begin{align}
K\left(S-\frac{N}{K}\right) &\geq \frac{1}{d} \sum_{i=1}^K \sum_{k\in [1:K]\setminus i} H\left(\A^{t+1}(i,k)\right) \overset{(a)}{=} \sum_{i=1}^K \sum_{k\in [1:K]\setminus i}\: \sum_{\mathcal{W} \subseteq [1:K]\setminus i :\: k \in \mathcal{W}} \vert \A^{t+1}_{\mathcal{W}}(i)\vert \nonumber \\
&\overset{(b)}{=}\sum_{i=1}^K \sum_{\mathcal{W} \subseteq [1:K]\setminus i} \vert \mathcal{W}\vert \: \vert \A^{t+1}_{\mathcal{W}}(i)\vert = \sum_{\ell=1}^K \ell \sum_{i=1}^K \sum_{\mathcal{W} \subseteq [1:K]\setminus i:\: \vert \mathcal{W}\vert =t\vert}  \vert \A^{t+1}_{\mathcal{W}}(i)\vert \overset{(c)}{=} \sum_{\ell=1}^K \ell x_{\ell},\label{eq:constraint_cache}
\end{align}
where $(a)$ follows from \eqref{eq:def_A_k_W}, $(b)$ is true because when we sum up the contents of the excess storage at all the workers, the chunk $\A^{t+1}_{\mathcal{W}}(i)$ is counted $\vert \mathcal{W}\vert$ number of times, which is the number of workers storing this chunk, and $(c)$ follows from \eqref{eq:x_ell}.

We next define the worst-case communication as follows: 
\begin{definition}[\textbf{Worst-Case Communication}] For any achievable scheme characterized by the encoding, decoding, and cache update functions $(\phi,\psi,\mu)$, the worst-case communication overhead over all possible consecutive data shuffles $(\pi_{t}, \pi_{t+1})$ is defined as
\begin{equation}
R_{\text{worst-case}}^{(\phi,\psi,\mu)}(S)\: \overset{\Delta}{=}\:\underset{(\pi_{t},\pi_{t+1})}{\max}\; R_{(\pi_{t}, \pi_{t+1})}^{(\phi,\psi,\mu)}(S).
\end{equation}
\end{definition}

Our goal in this work is to characterize the optimal worst-case communication $R_{\text{worst-case}}^*(K,N,S)$ defined as
\begin{equation}
R_{\text{worst-case}}^*(S)\: \overset{\Delta}{=}\:\underset{(\phi,\psi,\mu)}{\min} \;R_{\text{worst-case}}^{(\phi,\psi,\mu)}(S).
\end{equation}
We next present a claim which shows that the optimal worst-case communication $R_{\text{worst-case}}^*(S)$  is a convex function of the storage $S$:
\begin{claim}
\label{cl:1}
$R_{\text{worst-case}}^*(S)$ is a convex function of $S$, where $S$ is the available storage at each worker.
\end{claim}

\begin{proof}
Claim~\ref{cl:1} follows from a simple memory sharing argument which shows that for any two available storage values $S_1$ and $S_2$, if $(S_1,R_{\text{worst-case}}^*(S_1))$, and $(S_2,R_{\text{worst-case}}^*(S_2))$ are achievable optimal schemes, then for any storage $\bar{S}=\alpha S_1 +(1-\alpha) S_2$, $0\leq\alpha\leq1$, there is a scheme which achieves a communication overhead of $\bar{R}(\bar{S})=\alpha R_{\text{worst-case}}^*(S_1) +(1-\alpha) R_{\text{worst-case}}^*(S_2)$. 

This is done as follows: first, we divide the data-set $\A$ across $d$ dimensions into 2 batches namely; $\A^{(\alpha)}$, and $\A^{(1-\alpha)}$ of dimensions $\alpha d$, and $(1-\alpha) d$, for each point respectively. Then, we divide the storage for every worker $w_k$ into 2 parts namely; $Z_k^{(\alpha)}$, and $Z_k^{(1-\alpha)}$ of size $S_1\alpha d$, and $S_2(1-\alpha)d$, respectively. The former batch $\A^{(\alpha)}$ will be shuffled among the former part of the storage $Z_k^{(\alpha)}$ to achieve the point $(S_1,R_{\text{worst-case}}^*(S_1))$, while the latter batch $\A^{(1-\alpha)}$ will be shuffled among the latter part of the storage $Z_k^{(1-\alpha)}$ to achieve the point $(S_2,R_{\text{worst-case}}^*(S_2))$. Therefore, the total achievable load is given by
\begin{equation}
\label{eq:memory-sharing}
H(X)=R_{\text{worst-case}}^*(S_1) \alpha d+ R_{\text{worst-case}}^*(S_2) (1-\alpha) d= \bar{R}(\bar{S})d.
\end{equation}
We next note that the optimal communication rate $R_{\text{worst-case}}^*(\bar{S})$ is upper bounded by $\bar{R}(\bar{S})$, the rate of the memory sharing scheme, i.e.,
\begin{align}
R_{\text{worst-case}}^*(\alpha S_1 +(1-\alpha) S_2) \leq \alpha R_{\text{worst-case}}^*(S_1) +(1-\alpha) R_{\text{worst-case}}^*(S_2),
\end{align}
which shows that $R_{\text{worst-case}}^*(S)$ is a convex function of $S$.
\end{proof}

\subsection{Notation}
The notation ${\left[n_1:n_2\right]}$ for $n_1<n_2$, and $n_1,n_2\in \mathbb{N}$ represents the set of all integers between $n_1$, and $n_2$, i.e., $\left[n_1:n_2\right]={\{n_1,n_1+1,\ldots, n_2\}}$.
The combination coefficient ${{{n}\choose {k}}=\frac{n!}{(n-k)! k!}}$ equals zero for ${k>n}$, or $k<0$.
In order to describe subsets of ordered sets, we use the subscript to give the indexes of the elements being chosen from the set, e.g., for the ordered set $\pi=(\pi_1\,\ldots,\pi_n)$, $\pi_{[1:4]} = (\pi_1,\pi_2,\pi_3,\pi_4)$.
We denote Random Variables (RVs) by capital letters, ordered sets of RVs by capital bold letters, and sets of data points/sub-points by calligraphy letters.
The set in the subscript of a set of ordered RVs  is used for short notation of a subset of the set of RVs, e.g.,  for a set of RVs $\Z =\{Z_1,\ldots,Z_n\}$, we use $\Z_{\mathcal{W}}$ to denote the set $\{Z_i\}_{i\in\mathcal{W}}$.

For a data-set $\mathcal{A}$, we use the notation $\mathcal{A}^t(i)$ to denote the data partition assigned to the worker $w_i$ at iteration $t$, $\mathcal{A}^t(i,j)$, for $i\neq j$, to denote the part of $\mathcal{A}^t(i)$ which is stored in the excess storage of the worker $w_j$ at iteration $t$, while $\mathcal{A}^t_j(i)$, for $i\neq j$, to denote  the part of $\mathcal{A}^t(i)$ which is \textit{only} stored in excess storage of the worker $w_j$ at iteration $t$, i.e., $\mathcal{A}^t_j(i)=\mathcal{A}^t(i,j)\setminus \cup_{k\not\in\{i,j\}}\mathcal{A}^t(i,k)$.
The notation $\mathcal{A}^t(\mathcal{W})$ is used to denote the union of the data partitions assigned to the workers at iteration $t$ whose indexes are in the set $\mathcal{W}$, i.e., 
$\mathcal{A}^t(\mathcal{W}) = \cup_{i\in\mathcal{W}} \mathcal{A}^t(i)$. Similarly we have, $\mathcal{A}^t(\mathcal{W},j) = \cup_{i\in\mathcal{W}} \mathcal{A}^t(i,j)$, where $j\not\in\mathcal{W}$, and $\mathcal{A}^t(i,\mathcal{W})=\cup_{j\in \mathcal{W}}\mathcal{A}^t(i,j)$, where $i\not\in\mathcal{W}$.
The notation $\mathcal{A}^t_{\mathcal{W}}(i)$, where $i \not\in \mathcal{W}$, denotes a subset of $\mathcal{A}^t(i)$ which is \textit{exclusively and jointly} stored at iteration $t$ in the excess storage  of all the workers whose indexes are in the set $\mathcal{W}$, i.e., $\mathcal{A}^t_{\mathcal{W}}(i)= \cap_{j\in\mathcal{W}} \mathcal{A}^t(i,j) \setminus \cup_{j \not\in(\mathcal{W}\cup i)} \mathcal{A}^t(i,j)$.
The following table summarizes the notation used to denote the subsets of the data-set $\A$:
\begin{center}
 \begin{tabular}{|c | p{7.2cm}| c |} 
 \hline
 Notation & Description & Representation \\ [0.5ex] 
 \hline\hline
 $\mathcal{A}^t(i)$  &  The data partition assigned to $w_i$ at iteration $t$. & -\\ 
 \hline
 $\mathcal{A}^t(i,j)$, $i\neq j$  &  A subset of $\mathcal{A}^t(i)$ stored in the excess storage of $w_j$ at iteration $t$. &  -\\
 \hline
  $\mathcal{A}_j^t(i)$, $i\neq j$  &  A subset of $\mathcal{A}^t(i)$ stored \textit{only} in the excess storage of $w_j$ at iteration $t$. & $\mathcal{A}^t(i,j)\setminus \cup_{k\not\in\{i,j\}}\mathcal{A}^t(i,k)$\\ 
 \hline
$\mathcal{A}^t(\mathcal{W})$  &  The union of the data partitions assigned to every $w_i$ at iteration $t$, where $i\in\mathcal{W}$. & $ \cup_{i\in\mathcal{W}} \mathcal{A}^t(i)$\\
 \hline
 $\mathcal{A}^t(\mathcal{W},j)$,  &  The union of the sets $\mathcal{A}^t(i,j)$ for $i\in\mathcal{W}$. $j\not\in \mathcal{W}$ &  $\cup_{i\in\mathcal{W}} \mathcal{A}^t(i,j)$\\
 \hline
 $\mathcal{A}^t(i,\mathcal{W})$, $i\not\in \mathcal{W}$  &  The union of the sets $\mathcal{A}^t(i,j)$  for $j\in\mathcal{W}$.& $\cup_{j\in\mathcal{W}} \mathcal{A}^t(i,j)$\\
 \hline
  $\mathcal{A}^t_{\mathcal{W}}(i)$, $i\not\in \mathcal{W}$  &  A subset of $\mathcal{A}^t(i)$ which is \textit{exclusively and jointly} stored in the excess storage at iteration $t$ of all the workers whose indexes are in the set $\mathcal{W}$. & $\cap_{j\in\mathcal{W}} \mathcal{A}^t(i,j) \setminus \cup_{j \not\in(\mathcal{W}\cup i)} \mathcal{A}^t(i,j)$\\ [1ex] 
 \hline
\end{tabular}
\end{center}

\section{Main Results and Discussions}
\label{sec:results}
The first theorem presents an achievable worst-case rate $R_{\text{worst-case}}$, which also yields an upper bound on the optimal storage-rate trade-off $R_{\text{worst-case}}^*$.

\begin{theorem}
\label{thm1}
For a data-set containing $N\in\mathbb{N}$ data points, and a set of $K\in\mathbb{N}$ distributed workers, the lower convex envelope of the following $K+1$ storage-rate pairs is achievable:
\begin{align}
\left(S=\left(1+i\frac{K-1}{K}\right)\frac{N}{K},\:R_{\text{worst-case}}^{\text{upper}}= \frac{N(K-i)}{K(i+1)}\right), \quad \forall i\in[0:K].
\end{align}
\end{theorem}
The proof of Theorem~\ref{thm1} is presented in Appendix~\ref{sec:upper-bound}.
We present an encoding, decoding, and cache update scheme, which achieves the above rate-storage pairs.
One of the crucial steps in the proof is the \textit{structural invariant placement and update} of the storage of the workers over time. The storage placement involves partitioning the data points across dimensions, which allows each worker to store  at least some parts of each data point, which in turns introduces a local storage gain for any potential data assignment.  In order to increase the global gain through increasing the coding opportunities, we minimize the overlap between the parts stored by each worker of each data point. Through a careful novel update of storage across time, the structure can be maintained for any random data assignment, which allows applying a coded data delivery mechanism to reduce the communication overhead. 
Now, we give the following illustrative example for $K=N=4$ to introduce the main elements of the achievability proof.

\begin{example} \normalfont
\label{ex1}
Consider the case of $K=4$ workers, and $N=4$ i.i.d. data points, labeled as $\{D_1,D_2,D_3,D_4\}$.
According to Theorem~\ref{thm1}, the achievable worst-case storage-rate trade-off is given by the lower convex envelope of the 5 storage-rate pairs ($S=3i/4+1$, $R =(4-i)/(i+1)$) for $i\in[0:4]$, which is also shown by the red curve in Figure~\ref{fig:ex_gap}.
\begin{figure}[t]
  \begin{center}
  \includegraphics[width=0.5\columnwidth]{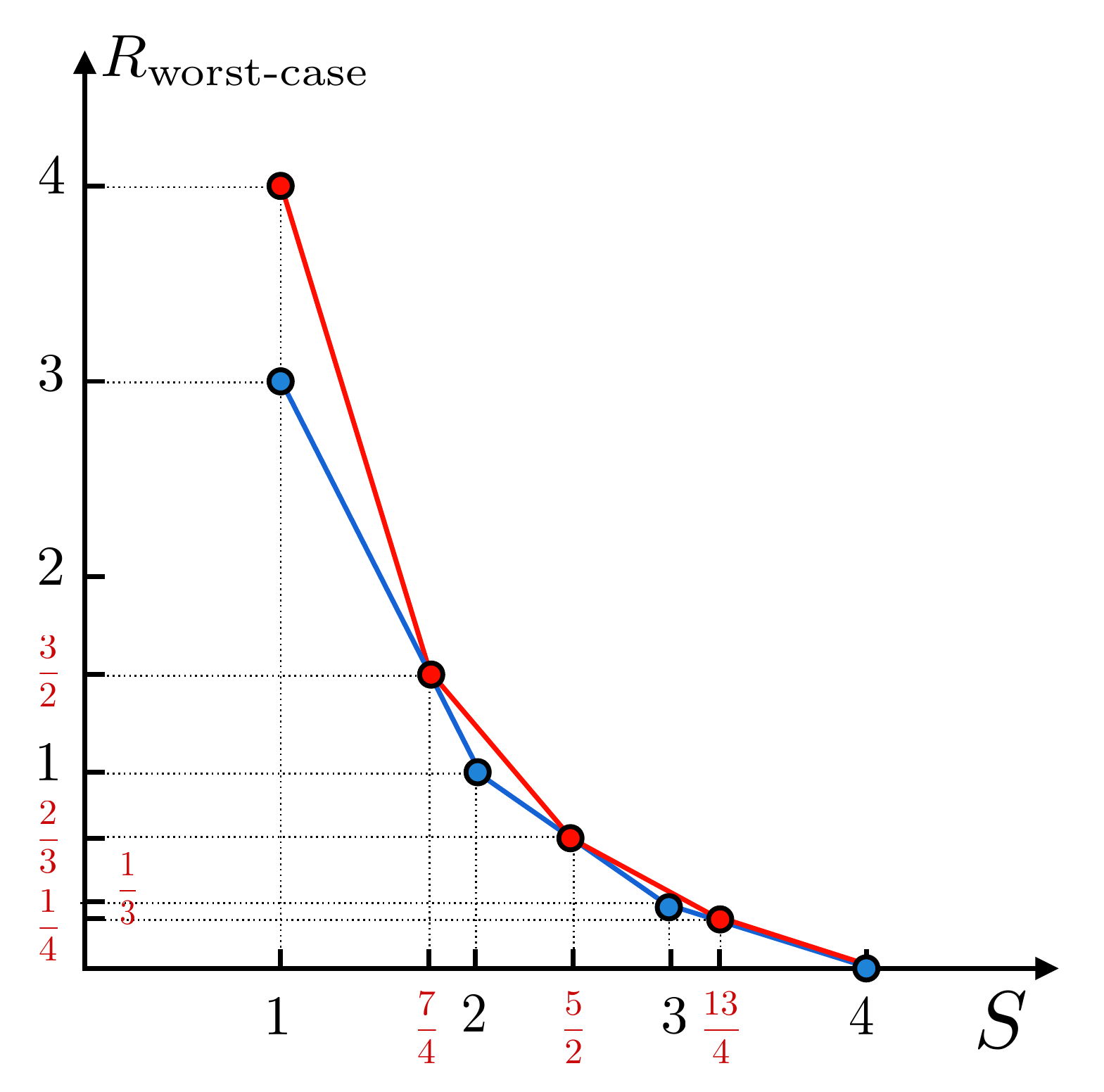}
\caption{The lower bound and the upper bound on the worst-case rate $R_{\text{worst-case}}^*$ for $N=4$, and $K=4$ versus the amount of storage $S$. The maximum gap appears to be when $S=1$, which is given as a ratio $4/3$. \label{fig:ex_gap}}
  \end{center}
\end{figure}
From Claim~\ref{cl:1}, once we achieve these pairs, the lower convex envelope is also achievable by memory sharing.
At time $t$, we consider the data is assigned according to the shuffle $\pi_t = (1,2,3,4)$, e.g., $w_1$ is assigned the data point $D_1$, i.e., $\A^t(1) = D_1$. 
At time $t+1$, we consider the cyclic shuffle $\pi_{t+1} = (2,3,4,1)$, e.g., $w_1$ is assigned the data point $D_2$ at time $t+1$, i.e., $\A^{t+1}(1) = D_2$.
Once we achieve the rate for the shuffle $\pi_{t+1} = (2,3,4,1)$, a similar data delivery mechanism can be used for any $\pi_{t+1}\in [4!]$, where $[4!]$ is the set containing all the $4!$ possible permutations of the set $[1:4]$.
The achievability, according to  $(\pi_t,\pi_{t+1})$, for the storage value $S = 3i/4+1$ and $i\in[0:4]$ follows next.

\noindent$\bullet$\hspace{5pt} \textbf{Case} $\mathbf{i=0}$ ($\mathbf{S=1}$): 

This storage value represents the no-excess storage case, where every worker only stores the assigned data point under processing. To satisfy the new assignment at time $t+1$, we choose now to send the $4$ data points, which satisfies any shuffle at time $t+1$, achieving the pair $(S=1,R=4)$. Later in Section~\ref{sec:close-gap}, we will show how to improve this rate and prove that in fact $(S=1,R=3)$ is optimal. The storage update is trivial in this case, where every worker keeps the new assigned data point and discard the remaining three points.

\begin{figure}[t]
  \begin{center}
  \includegraphics[width=0.99\columnwidth]{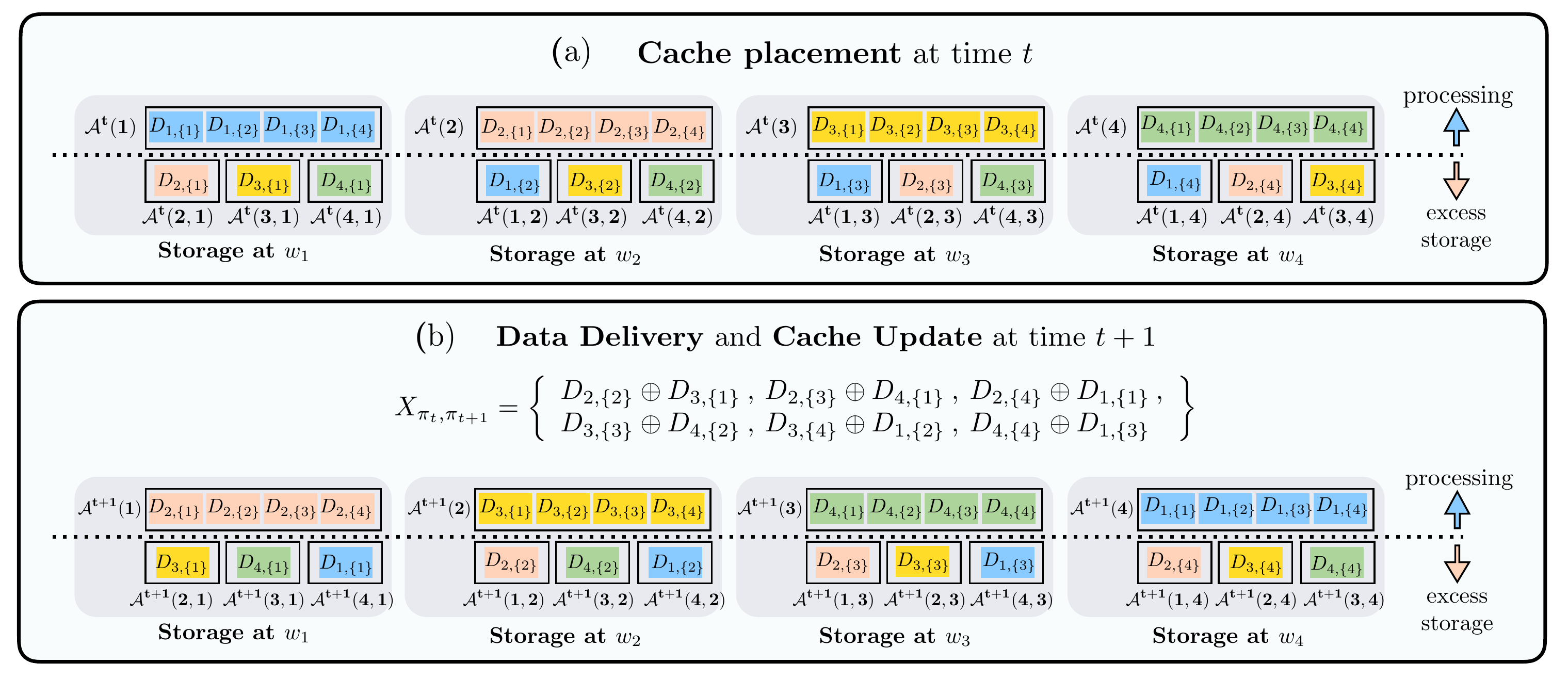}
\caption{Structural Invariant Storage placement, (a), and update, (b), for $K=4$ workers, $N=4$ data points, and $i=1$ ($S=7/4$).
Every data point is partitioned into $4$ sub-points each labeled by a unique subset of the set $[1:4]$ of length $1$. Above the dotted line is the data point fully stored for processing, and below the dotted line is the excess storage used to store the sub-points containing the worker's index.  \label{fig:ex-i-1}}
  \end{center}
\end{figure}

\noindent$\bullet$\hspace{5pt} \textbf{Case} $\mathbf{i=1}$ ($\mathbf{S=7/4}$):

\noindent \underline{Storage Placement:} The storage placement for $i=1$ is shown in Figure~\ref{fig:ex-i-1}a.
First, every data point is partitioned into $4$ sub-points of size $d/4$ bits each, where every sub-point is labeled by a unique subset $\mathcal{W}\subseteq [1:4]$ of size $\vert \mathcal{W}\vert =1$. For instance, the data point $D_1$ is partitioned as follows:
\begin{align}
D_1= \{D_{1,\{1\}},D_{1,\{2\}},D_{1,\{3\}},D_{1,\{4\}}\}.
\end{align}
Every worker first fully stores the assigned data point.
For the excess storage, every worker $w_k$ stores from the remaining points, not being processed, the sub-points where $k\in\mathcal{W}$. For instance, $w_1$ stores $1$ sub-point of $D_2$, labeled as $\A^t({2},1) =\{D_{2,\{1\}}\}$.
To summarize, each worker stores the assigned data point of size $d$, and for each one of the remaining $3$ data points, it stores $1$ sub-point of size $d/4$. That is, the storage requirement is given by
$S = 1 + 3\times 1/4= 7/4$,
which satisfies the storage constraint for $i=1$ ($S=7/4$).

\noindent \underline{Data Delivery:} According to the storage placement at time $t$ in Figure~\ref{fig:ex-i-1}a, at time $t+1$ every worker needs $3$ sub-points of the assigned data point, and every sub-point is available at least in one of the remaining workers, e.g., $w_1$ needs the sub-points $\{D_{2,\{2\}},D_{2,\{3\}},D_{2,\{4\}}\}$. Now, if we pick any $2$ out of the $4$ workers, then each one of the $2$ workers needs a sub-point available at the other worker. Therefore, we can send an \textit{``order $2$"} symbol, of size $d/4$ bits, useful for these chosen two workers in the same time, and for all possible choices of $2$ out of the $4$ workers we send the following $\binom{4}{2}=6$ coded symbols which satisfies the required $3$ needed sub-points for the $4$ workers:
\begin{align}
X_{\pi_t,\pi_{t+1}} = \left.\begin{cases}
D_{2,\{2\}}\oplus D_{3,\{1\}},  & \text{useful for $w_1,w_2$},\\
D_{2,\{3\}}\oplus D_{4,\{1\}},  &\text{useful for $w_1,w_3$},\\
D_{2,\{4\}}\oplus D_{1,\{1\}},  &\text{useful for $w_1,w_4$},\\
D_{3,\{3\}}\oplus D_{4,\{2\}},  &\text{useful for $w_2,w_3$},\\
D_{3,\{4\}}\oplus D_{1,\{2\}},  &\text{useful for $w_2,w_4$},\\
D_{4,\{4\}}\oplus D_{1,\{3\}},  &\text{useful for $w_3,w_4$}\\
\end{cases}\right\}.
\end{align}
The rate of this transmission is $\binom{4}{2}/4 = 3/2$, and the pair $(S=7/4, R = 3/2)$ is achieved.

\noindent \underline{Storage Update:} At time $t+1$, the storage update follows from  Figure~\ref{fig:ex-i-1}b. In order to maintain the structure of the storage, the workers first store the data points newly assigned and acquired from the delivery phase. For the excess storage update, each worker $w_k$ keeps  from the data point previously assigned at time $t$ the sub-points which are labeled by a set $\mathcal{W}$ where $k\in\mathcal{W}$. For example, $w_1$ keeps from $\A^t(1)=\A^{t+1}(4)=D_1$ the sub-point $\A^{t+1}(4,1)= \{D_{1,\{1\}}\}$.

\begin{figure}[t]
  \begin{center}
  \includegraphics[width=0.99\columnwidth]{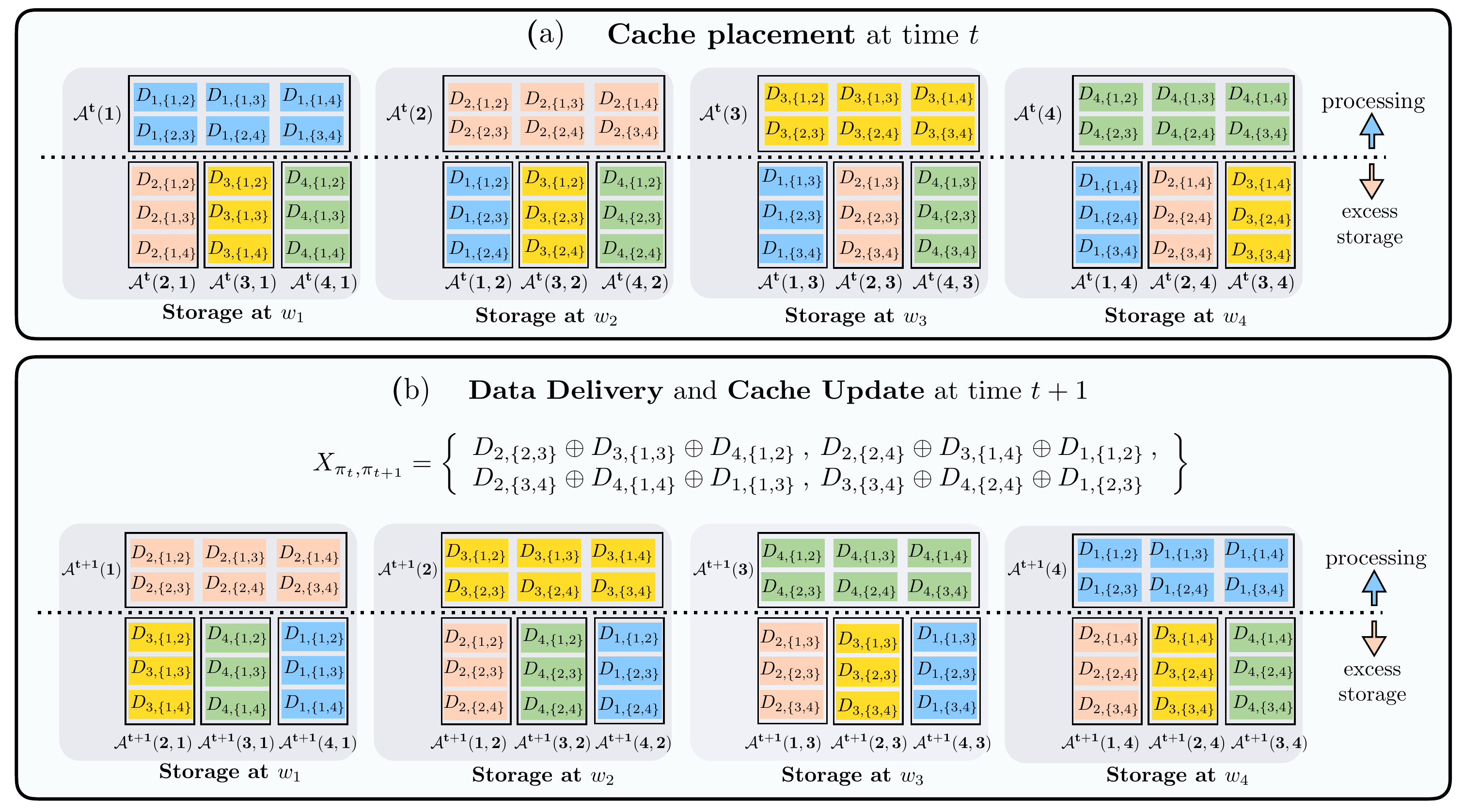}
\caption{Structural Invariant Storage placement, (a), and update, (b), for $K=4$ workers, $N=4$ data points, and $i=2$ ($S=5/2$).
Every data point is partitioned into $6$ sub-points each labeled by a unique subset of the set $[1:4]$ of length $2$. Above the dotted line is the data point fully stored for processing, and below the dotted line is the excess storage used to store the sub-points containing the worker's index.  \label{fig:ex-i-2}}
  \end{center}
\end{figure}

\noindent$\bullet$\hspace{5pt} \textbf{Case} $\mathbf{i=2}$ ($\mathbf{S=5/2}$): 

\noindent \underline{Storage Placement:} The storage placement for $i=2$ is shown in Figure~\ref{fig:ex-i-2}a.
First, every data point is partitioned into $6$ sub-points of size $d/6$ bits each, where every sub-point is labeled by a unique subset $\mathcal{W}\subseteq [1:4]$ of size $\vert \mathcal{W}\vert =2$. For instance, the data point $D_1$ is partitioned as follows:
\begin{align}
D_1= \{D_{1,\{1,2\}},D_{1,\{1,3\}},D_{1,\{1,4\}},D_{1,\{2,3\}},D_{1,\{2,4\}},D_{1,\{3,4\}}\}.
\end{align}
Every worker first fully stores the assigned data point.
For the excess storage, every worker $w_k$ stores from the remaining points, not being processed, the sub-points where $k\in\mathcal{W}$. For instance, $w_1$ stores $3$ sub-point of $D_2$, labeled as $\A^t({2},1) =\{D_{2,\{1,2\}},D_{2,\{1,3\}},D_{2,\{1,4\}}\}$.
To summarize, each worker stores the assigned data point of size $d$, and for each one of the remaining $3$ data points, it stores $3$ sub-point of size $d/6$ each. That is, the storage requirement is given by
$S = 1 + 3\times 3\times 1/6= 5/2$,
which satisfies the storage constraint for $i=2$ ($S=5/2$).

\noindent \underline{Data Delivery:} According to the storage placement at time $t$ in Figure~\ref{fig:ex-i-2}a, at time $t+1$ every worker needs $3$ sub-points of the assigned data point, and every sub-point is available at least in two of the remaining workers, e.g., $w_1$ needs the sub-points $\{D_{2,\{2,3\}},D_{2,\{2,4\}},D_{2,\{3,4\}}\}$. Now, if we pick any $3$ out of the $4$ workers, then every one of the $3$ workers needs a sub-point available at the other $2$ workers. Therefore, we can send an order $3$ symbol, of size $d/6$ bits, useful for these chosen workers in the same time, and for all possible choices of $3$ out of the $4$ workers we send the following $\binom{4}{3}=4$ coded symbols which satisfies the required $3$ needed sub-points for the $4$ workers:
\begin{align}
X_{\pi_t,\pi_{t+1}} = \left.\begin{cases}
D_{2,\{2,3\}}\oplus D_{3,\{1,3\}}\oplus D_{4,\{1,2\}},  & \text{useful for $w_1,w_2,w_3$},\\
D_{2,\{2,4\}}\oplus D_{3,\{1,4\}}\oplus D_{1,\{1,2\}},  &\text{useful for $w_1,w_2,w_4$},\\
D_{2,\{3,4\}}\oplus D_{4,\{1,4\}}\oplus D_{1,\{1,3\}},  &\text{useful for $w_1,w_3,w_4$},\\
D_{3,\{3,4\}}\oplus D_{4,\{2,4\}}\oplus D_{1,\{2,3\}},  &\text{useful for $w_2,w_3,w_4$}
\end{cases}\right\}.
\end{align}
The rate of this transmission is $\binom{4}{3}/\binom{4}{2} = 2/3$, and the pair $(S=5/2, R = 2/3)$ is achieved.

\noindent \underline{Storage Update:} At time $t+1$, the storage update follows from  Figure~\ref{fig:ex-i-2}b. In order to maintain the structure of the storage, the workers first store the data points newly assigned and acquired from the delivery phase. For the excess storage update, each worker $w_k$ keeps  from the data point previously assigned at time $t$ the sub-points which are labeled by a set $\mathcal{W}$ where $k\in\mathcal{W}$. For example, $w_1$ keeps from $\A^t(1)=\A^{t+1}(4)=D_1$ the sub-point $\A^{t+1}(4,1)= \{D_{1,\{1,2\}},D_{1,\{1,3\}},D_{1,\{1,4\}}\}$.

\begin{figure}[t]
  \begin{center}
  \includegraphics[width=0.99\columnwidth]{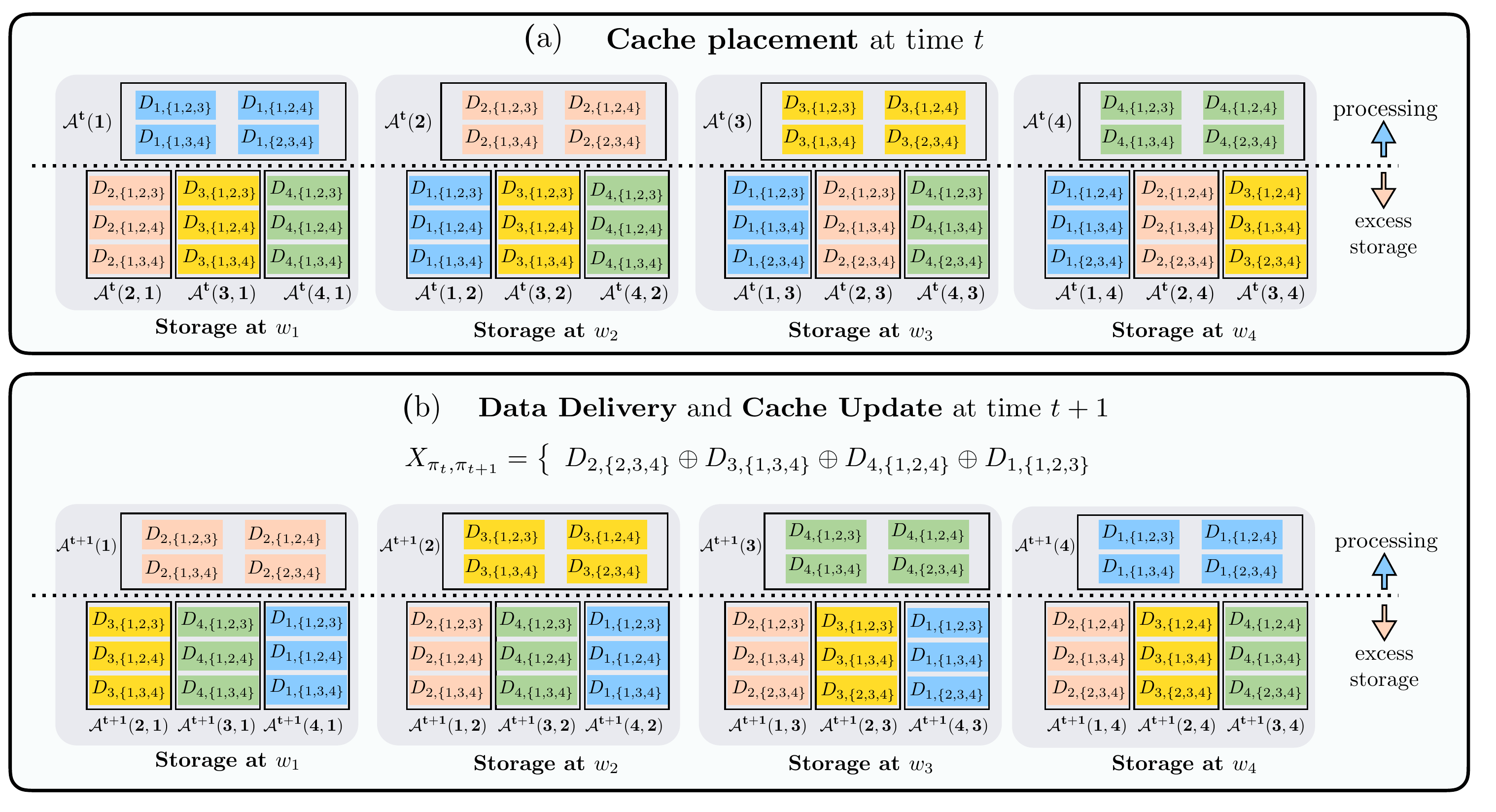}
\caption{Structural Invariant Storage placement, (a), and update, (b), for $K=4$ workers, $N=4$ data points, and $i=3$ ($S=13/4$).
Every data point is partitioned into $4$ sub-points each labeled by a unique subset of the set $[1:4]$ of length $3$. Above the dotted line is the data point fully stored for processing, and below the dotted line is the excess storage used to store the sub-points containing the worker's index.  \label{fig:ex-i-3}}
  \end{center}
\end{figure}

\noindent$\bullet$\hspace{5pt} \textbf{Case} $\mathbf{i=3}$ ($\mathbf{S=13/4}$): 

\noindent \underline{Storage Placement:} The storage placement for $i=3$ is shown in Figure~\ref{fig:ex-i-3}a.
First, every data point is partitioned into $4$ sub-points of size $d/4$ bits each, where every sub-point is labeled by a unique subset $\mathcal{W}\subseteq [1:4]$ of size $\vert \mathcal{W}\vert =3$. For instance, the data point $D_1$ is partitioned as follows:
\begin{align}
D_1= D_1= \{D_{1,\{1,2,3\}},D_{1,\{1,2,4\}},D_{1,\{1,3,4\}},D_{1,\{2,3,4\}}\}.
\end{align}
Every worker first fully stores the assigned data point.
For the excess storage, every worker $w_k$ stores from the remaining points, not being processed, the sub-points where $k\in\mathcal{W}$. For instance, $w_1$ stores $3$ sub-point of $D_2$, labeled as $\A^t({2},1) = \{D_{2,\{1,2,3\}},D_{2,\{2,2,4\}},D_{2,\{1,3,4\}}\}$.
To summarize, each worker stores the assigned data point of size $d$, and for each one of the remaining $3$ data points, it stores $3$ sub-point of size $d/4$ each. That is, the storage requirement is given by
$S = 1 + 3\times 3\times 1/4= 13/4$,
which satisfies the storage constraint for $i=3$ ($S=13/4$).

\noindent \underline{Data Delivery:} According to the storage placement at time $t$ in Figure~\ref{fig:ex-i-3}a, at time $t+1$ every worker only needs one sub-point of the assigned data point which is available at the three remaining workers, e.g., $w_1$ needs $D_{2,\{2,3,4\}}$ which is available at the workers $w_2$, $w_3$, and $w_4$. Therefore, we can send the following order $4$ symbol useful for all the $4$ workers at the same time:
\begin{align}
X_{\pi_t,\pi_{t+1}} = \{D_{2,\{2,3,4\}}\oplus D_{3,\{1,3,4\}}\oplus D_{4,\{1,2,4\}}\oplus D_{1,\{1,2,3\}}\}.
\end{align}
Since the the size of each sub-point is $d/4$, the rate of the transmission is $\binom{4}{4}/4 = 1/4$ and hence the pair $(S=13/4, R = 1/4)$ is achieved.

\noindent \underline{Storage Update:} At time $t+1$, the storage update follows from  Figure~\ref{fig:ex-i-3}b. In order to maintain the structure of the storage, the workers first store the data points newly assigned and acquired from the delivery phase. For the excess storage update, each worker $w_k$ keeps  from the data point previously assigned at time $t$ the sub-points which are labeled by a set $\mathcal{W}$ where $k\in\mathcal{W}$. For example, $w_1$ keeps from $\A^t(1)=\A^{t+1}(4)=D_1$ the sub-point $\A^{t+1}(4,1)= \{D_{1,\{1,2,3\}},D_{1,\{2,2,4\}},D_{1,\{1,3,4\}}\}$.

\noindent$\bullet$\hspace{5pt} \textbf{Case} $\mathbf{i=4}$ ($\mathbf{S=4}$):
This case is trivial where every worker can store all the $4$ data points and hence no communication is needed for any shuffle. Therefore, the pair $(S=4, R = 0)$ is achieved.

\vspace{10pt}

\hrule

\end{example}

We next present our second main result in Theorem \ref{thm2}, which gives an information theoretic lower bound on the optimal worst-case rate.

\begin{theorem}
\label{thm2}
For a data-set containing $N\in\mathbb{N}$ data points, and a set of $K\in\mathbb{N}$ distributed workers, a lower bound on  $R_{\text{worst-case}}^*$ is given by the lower convex envelope of the following $K$ storage-rate pairs:
\begin{align}
\left(S=m\frac{N}{K},\:R_{\text{worst-case}}^{\text{lower}}= \frac{N(K-m)}{Km}\right), \quad \forall m\in[1:K].
\end{align}
\end{theorem}
The complete proof of Theorem~\ref{thm2} can be found in Appendix~\ref{sec:lower-bound}.

\begin{remark}[\textbf{Basic idea for the converse}]\normalfont
\label{re:basic-idea-converse} A lower bound over the optimal rate $R_{\pi_t,\pi_{t+1}}^*$ of a shuffle $(\pi_t,\pi_{t+1})$ serves also as a lower bound on the worst-case since the optimal worst-case rate is larger than the optimal rate for any shuffle, i.e., $R_{\text{worst-case}}^*\geq R_{\pi_t,\pi_{t+1}}^*$. Therefore, we get lower bounds over $R_{\text{worst-case}}^*$ by focusing on a sequence of shuffles, and then average out all the lower bounds.
The novel part in our proof is to carefully choose the right shuffles which lead to the best lower bound.
\end{remark}
 In our converse proof, we apply a novel bounding methodology
similar to the recent result in \cite{Optimality/UncodedCache,Optimality/UncodedCache2}, where the optimal uncoded cache placement problem for a file delivery system is considered. 
In this paper however, we consider the data delivery based on subsequent assignments according to random shuffles of the data. Our problem also requires storing the data under processing, and allows for storage update over time as opposed to \cite{Optimality/UncodedCache,Optimality/UncodedCache2}.
At the end, we arrive at a linear program subject to the  problem constraints (data size and storage constraints), which can be solved to obtain the best lower bounds over different regimes of the available storage.
In the following example, we show how to obtain the lower bounds on the worst-case rate  for the case of $N=K=4$.

\begin{example}\normalfont
\label{ex2}
Consider the case of $K=4$ workers, and $N=4$ i.i.d. data points, labeled as $\{D_1,D_2,D_3,D_4\}$.
Assume the 4 data points are assigned at time $t$ according to $\pi_t=(1,2,3,4)$, i.e., $\A^t(k) = D_k$ for $k\in[1:4]$.
Therefore, at time $t$, the data point $D_k$ is fully stored at the cache of the worker $w_k$, and partially stored at the remaining workers, which gives the storage content of the worker $w_k$ as follows:
\begin{align}
\label{eq:ex:cache-content}
Z_k^{t+1} = \left\{D_k, \underset{j\in [1:4]\setminus k}{\cup}D_j(k)\right\},
\end{align}
where $D_j(k)$ is the part of $D_j$ stored in the excess storage of worker $w_k$ at time $t$.

Following Remark~\ref{re:basic-idea-converse}, we start by considering the following shuffle $(\pi_t,\pi_{t+1})$: for a permutation $\sigma:(1,2,3,4)\rightarrow (\sigma_1,\sigma_2,\sigma_3,\sigma_4)$, the worker $w_{\sigma_k}$ is assigned at time $t+1$ the data point
 that was assigned to the worker $w_{\sigma_{k-1}}$ at time $t$, i.e., $\A^{t+1}({\sigma_{k}}) =\A^{t}({\sigma_{k-1}})= D_{\sigma_{k-1}}$. Using the decodability constraint in (\ref{eq:decoding-const}), worker $w_{\sigma_{k}}$ must be able to decode $\A^{t+1}({\sigma_{k}}) = D_{\sigma_{k-1}}$ using its own cache $Z_{\sigma_{k}}^t$ as well as the transmission $X_{(\pi_t,\pi_{t+1})}$ which gives the following condition:
\begin{align}
\label{eq:ex:deocdability}
 H(\A^{t+1}({\sigma_{k}})| Z^{t}_{\sigma_{k}},X_{(\pi_t,\pi_{t+1})})= H(D_{\sigma_{k-1}}| Z^{t}_{\sigma_{k}},X_{(\pi_t,\pi_{t+1})})=0,\quad \forall k\in[1:4].
\end{align}
Furthermore, from (\ref{eq:cache-min-desired}), each worker should store the assigned data point at time $t$, therefore,
\begin{align}
H(\A^t(\sigma_{k})|Z_{\sigma_{k}}^t) = H(D_{\sigma_{k}}|Z_{\sigma_{k}}^t) =0, \quad \forall k\in[1:4].
\label{eq:ex:cache-min-desired}
\end{align}
Consequently, the transmission $X_{\pi_t,\pi_{t+1}}$ as well as the cache of any three workers can decode the $4$ data points, which can be shown  as follows:
\begin{align}
H(\A|Z_{\sigma_{[2:4]}}^t,X_{(\pi_t,\pi_{t+1})}) &= H(D_1,D_2,D_3,D_4|Z_{\sigma_{[2:4]}}^t,X_{(\pi_t,\pi_{t+1})}) \nonumber\\
& \hspace{-15pt}= H(D_{\sigma_1},D_{\sigma_2},D_{\sigma_3},D_{\sigma_4}|Z_{\sigma_{[2:4]}}^t,X_{(\pi_t,\pi_{t+1})}) \nonumber\\
&\hspace{-30pt}\overset{(a)}{\leq} 
H(D_{\sigma_1}|Z_{\sigma_2}^t,X_{(\pi_t,\pi_{t+1})})+H(D_{\sigma_2}|Z_{\sigma_2}^t)+H(D_{\sigma_3}|Z_{\sigma_3}^t)+H(D_{\sigma_4}|Z_{\sigma_4}^t)\overset{(b)}{=}0, \label{eq:ex:bound1}
\end{align}
where $(a)$ follows from the fact that $H(A,B)\leq H(A)+H(B)$ and that conditioning reduces entropy, and $(b)$ follows directly using \eqref{eq:ex:deocdability} and \eqref{eq:ex:cache-min-desired}.
Next, we obtain the following bound using \eqref{eq:ex:bound1}:
\begin{align}
4d &= H(\A) =I(A;\Z_{\sigma_{[2:4]}}^t,X_{\pi_t,\pi_{t+1}}) + H(\A|\Z_{\sigma_{[2:4]}}^t,X_{\pi_t,\pi_{t+1}})\nonumber\\
&\overset{(a)}{\leq} H(\Z_{\sigma_{[2:4]}}^t,X_{\pi_t,\pi_{t+1}})  \overset{(b)}{=} H(Z_{\sigma_4}^t,X_{\pi_t,\pi_{t+1}}) +H(Z_{\sigma_2}^t,Z_{\sigma_3}^t|Z_{\sigma_4}^t,X_{\pi_t,\pi_{t+1}})\nonumber\\
& \leq H(X_{\pi_t,\pi_{t+1}}) + H(Z_{\sigma_4}^t) +H(Z_{\sigma_3}^t|Z_{\sigma_4}^t,X_{\pi_t,\pi_{t+1}})+H(Z_{\sigma_2}^t|Z_{\sigma_3}^t,Z_{\sigma_4}^t,X_{\pi_t,\pi_{t+1}})\nonumber\\
& \overset{(c)}{\leq} R^*_{\pi_t,\pi_{t+1}}d +  H(Z_{\sigma_4}^t) + H(Z_{\sigma_3}^t|Z_{\sigma_4}^t,D_{\sigma_3},D_{\sigma_4})+H(Z_{\sigma_2}^t|Z_{\sigma_3}^t,Z_{\sigma_4}^t,D_{\sigma_2},D_{\sigma_3},D_{\sigma_4})\nonumber\\
&\overset{(d)}{=} R^*_{\pi_t,\pi_{t+1}}d +H(D_{\sigma_4},D_{\sigma_1}({\sigma_4}),D_{\sigma_2}({\sigma_4}),D_{\sigma_3}({\sigma_4})) + H(D_{\sigma_1}({\sigma_3}),D_{\sigma_2}({\sigma_3})|Z_{\sigma_4}^t)\nonumber\\
&\qquad  +H(D_{\sigma_1}({\sigma_2})|Z_{\sigma_3}^t,Z_{\sigma_4}^t)\nonumber\\
&\overset{(e)}{=} R^*_{\pi_t,\pi_{t+1}}d +H(D_{\sigma_4},D_{\sigma_1}({\sigma_4}),D_{\sigma_2}({\sigma_4}),D_{\sigma_3}({\sigma_4})) + H(D_{\sigma_1}({\sigma_3}),D_{\sigma_2}(\sigma_3)|D_{\sigma_1}({\sigma_4}),D_{\sigma_2}({\sigma_4})) \nonumber\\
&\qquad +H(D_{\sigma_1}({\sigma_2})|D_{\sigma_1}({\sigma_3}),D_{\sigma_1}({\sigma_4}))\nonumber\\
& = R^*_{\pi_t,\pi_{t+1}}d +H(D_{\sigma_4}) +[H(D_{\sigma_1}({\sigma_4}))+H(D_{\sigma_1}({\sigma_3})|D_{\sigma_1}({\sigma_4}))+H(D_{\sigma_1}({\sigma_2})|D_{\sigma_1}({\sigma_3}),D_{\sigma_1}({\sigma_4}))] \nonumber\\
&\qquad+[H(D_{\sigma_2}({\sigma_4}))+H(D_{\sigma_2}({\sigma_3})|D_{\sigma_2}({\sigma_4}))] +H(D_{\sigma_3}({\sigma_4}))\nonumber\\
& \overset{(f)}{=} R^*_{\pi_t,\pi_{t+1}}d +d +H(D_{\sigma_1}({\sigma_2},{\sigma_3},{\sigma_4})) +H(D_{\sigma_2}({\sigma_3},{\sigma_4})) +H(D_{\sigma_3}({\sigma_4}))\nonumber\\
&\overset{(g)}{\leq} R^*_{\text{worst-case}}d +d +H(D_{\sigma_1}({\sigma_2},{\sigma_3},{\sigma_4})) +H(D_{\sigma_2}({\sigma_3},{\sigma_4})) +H(D_{\sigma_3}({\sigma_4})),
\end{align}
where $(a)$ follows from \eqref{eq:ex:bound1}, \eqref{eq:cache-content}, and \eqref{eq:transmit-load}, where $\Z^t_{\sigma_{[2:4]}}$, and $X_{\pi_t,\pi_{t+1}}$ are deterministic functions of the data-set  $\A$, $(b)$ from the chain rule of entropy, $(c)$ follows from  \eqref{eq:ex:deocdability}, (\ref{eq:ex:cache-min-desired}), and because conditioning reduces entropy, $(d)$ follows from the storage content at time $t$ given in \eqref{eq:ex:cache-content}, where after knowing $\{D_{\sigma_3},D_{\sigma_4}\}$ (or similarly $\{D_{\sigma_2},D_{\sigma_3},D_{\sigma_4}\}$), the only parts left in $Z_{\sigma_3}^t$ (or $Z_{\sigma_2}^t$) are $\{D_{\sigma_1}({\sigma_3}),D_{\sigma_2}({\sigma_3})\}$ ($\{D_{\sigma_1}({\sigma_2})\}$), $(e)$ follows since out of the cache contents $Z_j^t$, the data sub-point $D_k(i)$ only depends on the sub-point $D_k(j)$, for any $i\neq j$, $(f)$ follows from the chain rule of entropy where $D_i(\mathcal{W})$ is the part of $D_i$ stored in the excess storage of the workers with the index $w_j$ where $j\in\mathcal{W}$ at time t, and finally $(g)$ follows from Remark~\ref{re:basic-idea-converse}.

Summing up over all possible $4! =24$ permutations of the ordered set $(1,2,3,4)$, we arrive at the following bound,
\begin{align}
R^*_{\text{worst-case}}d \geq 3d-\frac{1}{24}\sum_{\mathbf{\sigma}\in[4!]}\left[H(D_{\sigma_1}({\sigma_2,\sigma_3,\sigma_4})) +H(D_{\sigma_2}({\sigma_3,\sigma_4})) +H(D_{\sigma_3}({\sigma_4}))\right]\nonumber\\
\overset{(a)}{=}3d- \frac{1}{24}\sum_{\mathbf{\sigma}\in[4!]}\left[H(D_{\sigma_1}({\sigma_2,\sigma_3,\sigma_4})) +H(D_{\sigma_1}({\sigma_2,\sigma_3})) +H(D_{\sigma_1}({\sigma_2}))\right],
\label{eq:ex:bound2}
\end{align}
where $[4!]$ is the set of all possible permutations of the ordered set $(1,2,3,4)$, and $(a)$ follows due to the symmetry in the summation by simple change of summation indexes.
Following  the definition in \eqref{eq:def_A_k_W}, we can define $D_{k,\mathcal{W}}$ as the part of $D_k$ stored exclusively in the excess storage of the workers whose labels are in the set $\mathcal{W}$. 
According to $\pi_t =(1,2,3,4)$, at time $t$, $D_{k,\mathcal{W}}$ is only defined for $k\not\in \mathcal{W}$ ($w_k$ does not store $D_k$ as excess storage).
Therefore, at time $t$, we can express the following entropies in terms of $D_{k,\mathcal{W}}$ as follows:
\begin{align}
H(D_k) = \sum_{\mathcal{W}\in 2^{[1:4]\setminus k}} \vert D_{k,\mathcal{W}}\vert d,\quad H(D_k(j)) = \sum_{\substack{\mathcal{W} \subseteq [1:K]\setminus k\\ j\in \mathcal{W}}} \vert D_{k,\mathcal{W}}\vert d,
\end{align}
where $ \vert D_{k,\mathcal{W}}\vert $ is entropy of $D_{k,\mathcal{W}}$ normalized by the data point size $d$.
In the summation term of \eqref{eq:ex:bound2}, we obtain the term $\vert D_{k,\mathcal{W}}\vert$ only for $\vert\mathcal{W}\vert\in\{1,2,3\}$. Next, we show how to find the coefficients of $\vert D_{k,\mathcal{W}}\vert$ for different sizes of $\mathcal{W}$.

\noindent$ \bullet$ \hspace{5pt}\textbf{Coefficient of $\vert D_{k,\mathcal{W}}\vert$ for $\vert \mathcal{W} \vert =1$:}
Due to symmetry, we notice that obtaining the coefficient of $\vert D_{k,\mathcal{W}}\vert$ in the summation in \eqref{eq:ex:bound2} for any $\vert \mathcal{W} \vert =1$; is equivalent to obtaining the coefficient of $\vert D_{1,\{2\}}\vert$. We get  $\vert D_{1,\{2\}}\vert$ in the first term of the summation, i.e., $H(D_{\sigma_1}({\sigma_2,\sigma_3,\sigma_4}))$ only if $\sigma_1 =1$ which is satisfied in $6$ out of the $24$ permutations. In the second term, i.e., $H(D_{\sigma_1}(\sigma_2,\sigma_3))$, we obtain $\vert D_{1,\{2\}}\vert$ only if  $\sigma_1=1$ and $\sigma_4\neq 2$ in total number of 4 permutations. In the third term,  i.e., $H(D_{\sigma_1}(\sigma_2))$, we obtain $\vert D_{1,\{2\}}\vert$ only if  $\sigma_1=1$ and $\sigma_2= 2$ in total number of $2$ permutations. Therefore, the coefficient of $\vert D_{1,\{2\}}\vert$, hence any $\vert D_{k,\mathcal{W}}\vert$ for $\vert \mathcal{W} \vert =1$, is $\frac{6+4+2}{24}=\frac{1}{2}$.

\noindent$ \bullet$ \hspace{5pt}\textbf{Coefficient of $\vert D_{k,\mathcal{W}}\vert$ for $\vert \mathcal{W} \vert =2$:}
Similarly, we obtain the coefficient of $\vert D_{k,\mathcal{W}}\vert$ for any $\vert \mathcal{W} \vert =2$ by obtaining the coefficient of $\vert D_{1,\{2,3\}}\vert$. We get  $\vert D_{1,\{2,3\}}\vert$ in the first two terms of the summation only if $\sigma_1 =1$ which is satisfied in $6$ out of the $24$ permutations. In the third term, we obtain  $\vert D_{1,\{2,3\}}\vert$ only if  $\sigma_1=1$ and $\sigma_2\in \{2,3\}$ in total number of $4$ permutations. Therefore, the coefficient of $\vert D_{1,\{2,3\}}\vert$, hence any $\vert D_{k,\mathcal{W}}\vert$ for $\vert \mathcal{W} \vert =2$, is $\frac{6+6+4}{24}=\frac{2}{3}$.

\noindent$ \bullet$ \hspace{5pt}\textbf{Coefficient of $\vert D_{k,\mathcal{W}}\vert$ for $\vert \mathcal{W} \vert =3$:}
We obtain the coefficient of $\vert D_{k,\mathcal{W}}\vert$ for any $\vert \mathcal{W} \vert =3$ by obtaining the coefficient of $\vert D_{1,\{2,3,4\}}\vert$. We get $\vert D_{1,\{2,3,4\}}\vert$ in the first three terms of the summation only if $\sigma_1 =1$ which is satisfied in $6$ out of the $24$ permutations.  Therefore, the coefficient of $\vert D_{1,\{2,3,4\}}\vert$, hence any $\vert D_{k,\mathcal{W}}\vert$ for $\vert \mathcal{W} \vert =3$, is $\frac{6+6+6}{24}=\frac{3}{4}$.

Therefore, we can simplify the bound in \eqref{eq:ex:bound2} as follows:
\begin{align}
R^*_{\text{worst-case}}d \geq 3d-\frac{1}{2} \sum_{k=1}^4 \sum_{\substack{\mathcal{W} \subseteq [1:K]\setminus k\\ \vert\mathcal{W}\vert =1}} &\vert D_{k,\mathcal{W}}\vert d -\frac{2}{3} \sum_{i=1}^4 \sum_{\substack{\mathcal{W} \subseteq [1:K]\setminus k\\ \vert\mathcal{W}\vert =2}} \vert D_{k,\mathcal{W}}\vert d-\frac{3}{4} \sum_{k=1}^4 \sum_{\substack{\mathcal{W} \subseteq [1:K]\setminus k\\ \vert\mathcal{W}\vert =3}} \vert D_{k,\mathcal{W}}\vert d\nonumber\\
&= 3d -\frac{x_1d}{2} - \frac{2x_2 d}{3} - \frac{3x_3d}{4},
\label{eq:ex:bound3}
\end{align}
where $x_{\ell}$ for $\ell\in[0:3]$ is defined similar to \eqref{eq:x_ell} as $
 x_{\ell}=\sum_{k=1}^K \sum_{\substack{\mathcal{W}\subseteq [1:4]\setminus k:\: \vert\mathcal{W}\vert =\ell}}\vert D_{k,\mathcal{W}}\vert.$
By dividing both sides by $d$, we get the following bound:
\begin{align}
R^*_{\text{worst-case}}\geq 3 -\frac{x_1}{2} - \frac{2x_2}{3}  - \frac{3x_3}{4}.
\label{eq:ex:bound4}
\end{align} 
Moreover, the data size  and the excess storage size constraints for this example follow \eqref{eq:size-const} and \eqref{eq:constraint_cache}, respectively. Hence, we obtain the following constraints:
\begin{align}
&x_0+x_1+x_2+x_3 =4,\label{ex:size-const}\\
&x_1+2x_2+3x_3 \leq 4(S-1).\label{ex:cache-const}
\end{align}

We get the first bound over $R^*_{\text{worst-case}}$  by eliminating $x_1$ from \eqref{eq:ex:bound4} using the bound in \eqref{ex:cache-const} as follows:
\begin{align}
R^*_{\text{worst-case}} &\geq 3 -\frac{x_1}{2} - \frac{2x_2}{3}  - \frac{3x_3 }{4}\geq
3-\frac{1}{2}\left(4(S-1) -2x_2-3x_3\right) - \frac{2x_2}{3}  - \frac{3x_3}{4}\nonumber\\
&=5 -2S+\frac{x_2 }{3}+\frac{3x_3}{4} \overset{(a)}{\geq} 5 -2S,\label{eq:ex:bound-fa}
\end{align}
where $(a)$ follows since $x_2,x_3\geq 0$.

We get the second bound over $R^*_{\text{worst-case}}$ in two steps. First, we eliminate $x_1$ from \eqref{eq:ex:bound4} and \eqref{ex:cache-const} using \eqref{ex:size-const} to get the following two bounds:
\begin{align}
R^*_{\text{worst-case}} &\geq 3 -\frac{1}{2}\left(4-x_0-x_2-x_3\right) - \frac{2x_2}{3}  - \frac{3x_3}{4}  = 1 +\frac{x_0 }{2}-\frac{x_2}{6}-\frac{x_3}{4},\label{eq:ex:bound6a}\\
4(S-1) &\geq \left(4-x_0-x_2-x_3\right) +2x_2 +3x_3 = 4-x_0 +x_2+2x_3 .\label{eq:ex:bound6b}
\end{align}
We eliminate $x_2$ from \eqref{eq:ex:bound6a} using the bound in \eqref{eq:ex:bound6b} to obtain
\begin{align}
R^*_{\text{worst-case}} &\geq  1 +\frac{x_0}{2} -\frac{x_2}{6}-\frac{x_3}{4} \geq 1 +\frac{x_0}{2} -\frac{1}{6}\left(4(S-1) -4 +x_0-2x_3\right)-\frac{x_3}{4}\nonumber\\
& = \frac{7}{3} - \frac{2S}{3}+\frac{x_0}{3}+\frac{x_3}{12}\overset{(a)}{\geq}\frac{7-2S}{3},\label{eq:ex:bound-fb}
\end{align}
where $(a)$ follows since $x_0,x_3\geq 0$.

Following similar steps, we get a third bound over $R^*_{\text{worst-case}}$ by first eliminating $x_2$ from \eqref{eq:ex:bound4} and \eqref{ex:cache-const} using \eqref{ex:size-const} to get the following two bounds:
\begin{align}
R^*_{\text{worst-case}} &\geq 3 -\frac{x_1}{2} - \frac{2}{3} \left(4-x_0-x_1-x_3\right) - \frac{3x_3}{4}  = \frac{1}{3} +\frac{2x_0 }{3}+\frac{x_1}{6}-\frac{x_3}{12},\label{eq:ex:bound7a}\\
4(S-1)&\geq x_1+2\left(4-x_0-x_1-x_3\right)  +3x_3 = 8-2x_0 -x_+x_3.\label{eq:ex:bound7b}
\end{align}
We eliminate $x_3$ from \eqref{eq:ex:bound7a} using the bound in \eqref{eq:ex:bound7b} and arrive to
\begin{align}
R^*_{\text{worst-case}} &\geq   \frac{1}{3} +\frac{2x_0}{3} +\frac{x_1}{6}-\frac{x_3}{12} \geq \frac{1}{3} +\frac{2x_0 }{3}+\frac{x_1}{6}-\frac{1}{12}\left( 4(S-1) -8 +2x_0+x_1 \right)\nonumber\\
& = \frac{4}{3} - \frac{S}{3}+\frac{5x_0}{6}+\frac{x_1}{12}\overset{(a)}{\geq}\frac{4-S}{3},\label{eq:ex:bound-fc}
\end{align}
where $(a)$ follows since $x_0,x_1\geq 0$.

In summary, we obtain in \eqref{eq:ex:bound-fa}, \eqref{eq:ex:bound-fb}, and \eqref{eq:ex:bound-fc} the following bounds on $R^*_{\text{worst-case}}$:
\begin{align}
R^*_{\text{worst-case}} \geq 5-2S, \quad
R^*_{\text{worst-case}} \geq\frac{7-2S}{3},\quad 
R^*_{\text{worst-case}} \geq\frac{4-S}{3}.
\end{align}
The intersection of the three bounds is the lower convex hull of the $4$ storage-rate pairs, $(S= m, R= \frac{4-m}{m})$ for $m\in[1:4]$, which is the obtained lower bound over $R^*_{\text{worst-case}}$ given by the blue curve in Figure~\ref{fig:ex_gap}, satisfying Theorem~\ref{thm2} for $K=N=4$.
\vspace{10pt}
\hrule

\end{example}

In our next result, we compare the upper and lower bounds in Theorems \ref{thm1} and \ref{thm2}, respectively, and show that they are within a constant multiplication gap of each other.

\begin{theorem}
\label{thm3}
For a data-set containing $N\in\mathbb{N}$ data points, and a set of $K\in\mathbb{N}$ distributed workers, the gap ratio between the upper and the lower bounds on $R_{\text{worst-case}}^*$ given by Theorems~\ref{thm1}, and \ref{thm2}, respectively, is bounded as follows:
\begin{align}
\frac{R_{\text{worst-case}}^{\text{upper}}}{R_{\text{worst-case}}^{\text{lower}}} \leq \frac{K}{K-1} \leq 2.
\label{eq:thm3}
\end{align}
\end{theorem}
This Theorem shows that there is  a vanishing gap between the bounds as the number of workers $K$ increases, i.e., $\lim_{K\rightarrow \infty} \big(\frac{K}{K-1}\big) =1$. We also show that for the discrete set of storage points considered in Theorem~\ref{thm1}, i.e., $S=\left(1+i\frac{K-1}{K}\right) \frac{N}{K}$ for $i\in[1:K]$, our achievable scheme in fact optimal, and that the gap only results in the values of storage in between, i.e., memory sharing is not optimal in this case. For example, consider the bounds on $R_{\text{worst-case}}^*$ for $K=N=4$ shown in Figure~\ref{fig:ex_gap}. We first notice that the achieved storage-rate pairs $(S=7/4,R=3/2)$, $(S=5/2,R=2/3)$, and $(S=13/4,R=1/4)$ are optimal. Furthermore, we can show that the maximum gap between the bounds is at $S=1$, which is given by $4/3$, satisfying the bound in \eqref{eq:thm3}. The formal proof for the maximum gap analysis for any value of $K$ and $N$ can be found in Appendix~\ref{sec:gap-analysis}.

The next Theorem provides an improved gap through a new achievable scheme, which we call as \textit{``aligned coded shuffling"}.
\begin{theorem}
\label{thm4}
For a data-set containing $N\in\mathbb{N}$ data points, and a set of $K\in\mathbb{N}$ distributed workers, the lower bound over $R_{\text{worst-case}}^*$ in Theorem~\ref{thm2} is in fact achievable for $K<5$ (hence gives the optimal rate), while for $K\geq 5$ is achievable within a gap ratio bounded as
\begin{align}
\frac{R_{\text{worst-case}}^{\text{upper}}}{R_{\text{worst-case}}^{\text{lower}}} \leq \frac{K-\frac{1}{3}}{K-1} \leq \frac{7}{6}.
\end{align}
\end{theorem}
The above theorem is proved by closing the gap between the two bounds in Theorems~\ref{thm1} and \ref{thm2} for the storage values $S=m \frac{N}{K}$, and $m\in\{1,K-2,K-1\}$. This can be done with the use of sophisticated interference alignment mechanisms, which force the interference seen by each worker to occupy the minimum possible dimensions. In Appendix~\ref{sec:close-gap}, we present the complete proof of Theorem~\ref{thm4}.
To illustrate the new ideas introduced here, we revisit again Example~\ref{ex1} of $K=N=4$ to show how the gap between the lower and the upper bounds on $R^*_{\text{worst-case}}$ can be closed in this case.
 
 \begin{example}\normalfont
 \label{ex3}
 From Figure~\ref{fig:ex_gap}, we notice that if we close the gap for the storage points $S=m$, for $m\in[1:3]$, then we can fully characterize $R^*_{\text{worst-case}}$ using memory sharing between the achievable points (see Claim~\ref{cl:1}).  
In the achievability, we consider a different placement strategy, which is also invariant in the structure.  We also consider the aligned coded shuffling scheme for data delivery, which reduces the rate by forcing the interference to occupy the minimum possible dimensions.
 We consider the same subsequent shuffles $\pi_t =(1,2,3,4)$, and $\pi_{t+1} =(2,3,4,1)$. Furthermore, we define $\delta_t(i)$ as the index of the worker being assigned the data point $D_i$ at time $t$. Therefore, $\delta_t =(4,1,2,3)$, and $\delta_{t+1} =(1,2,3,4)$.
Next, we discuss the achievability for storage values $S=m$, and $m\in[1:3]$. 

\noindent$\bullet$\hspace{5pt} \textbf{Case} $\mathbf{m=1}$ ($\mathbf{S=1}$):

As mentioned before in Example~\ref{ex1}, the storage placement for the case $m=1$ (no excess storage) is trivial where every worker only stores the data point which needs to be processed.
We start by sending $3$ independent linear combinations of the $4$ data points as follows:
\begin{align}
X_{\pi_t,\pi_{t+1}} =\{L_1(D_1,D_2,D_3,D_4),\: L_2(D_1,D_2,D_3,D_4),\: L_3(D_1,D_2,D_3,D_4)\},
\end{align}
where $L_1$, $L_2$, and $L_3$ are three independent linear functions.
We notice that each worker already stores one data point, and then can decode the $3$ remaining data points and acquire the one needed at time $t+1$. For instance, worker $w_1$ has $D_1$ from the previous shuffle at time $t$, and then can get 3 independent linear functions enough to decode $D_2$, $D_3$, and $D_4$.
Therefore, the pair $(S=1,R=3)$ is achievable for $K=N=4$ closing the gap in Figure~\ref{fig:ex_gap} for $S=1$.
The storage update is also trivial, where every worker keeps the new assigned data point and discard the remaining three points. 

\begin{figure}[t]
  \begin{center}
  \includegraphics[width=0.99\columnwidth]{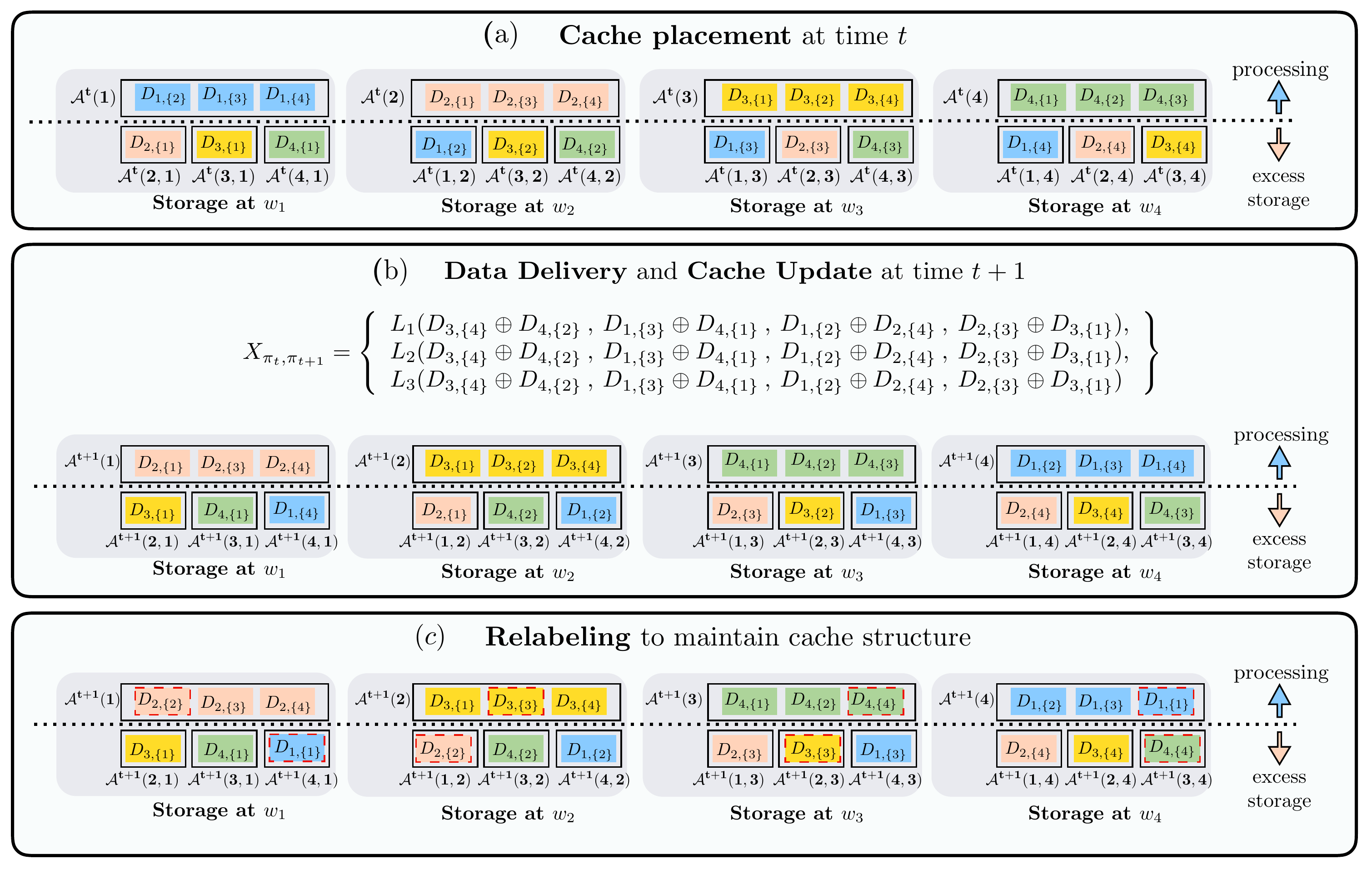}
\caption{
An example on closing the gap of $K=4$ workers, $N=4$ data points, and $m=2$ ($S=2$): (a) Structural Invariant Storage placement, (b) Data Delivery and storage update, and (c) Relabeling some sub-points in red dashed frames to maintain the storage structure.
At time $t$, every data point $D_i$ is partitioned into $3$ sub-points each labeled by a unique subset of length $1$ of the set $[1:4]\setminus \delta_t(i)$. 
At time $t+1$, for every data point $D_i$ the sub-point $D_{i,\{\delta_{t+1}(i)\}}$ is relabeled as $D_{i,\{\delta_{t}(i)\}}$.
Above the dotted line is the data point fully stored for processing, and below the dotted line is the excess storage used to store the sub-points containing the worker's index.\label{fig:ex-m-2}}
  \end{center}
\end{figure}

\noindent$\bullet$\hspace{5pt} \textbf{Case} $\mathbf{m=2}$ ($\mathbf{S=2}$):

\noindent \underline{Storage Placement} Every data point at time $t$ is partitioned into $3$ sub-points of size $d/3$ bits each, where every sub-point of the data point $D_i$  is labeled by a unique subset $\mathcal{W} \subseteq [1:4]\setminus \delta_t(i)$, where $\vert \mathcal{W} \vert =1$. For example, the data point $D_1$ at time $t$ is partitioned as $D_1= \{D_{1,\{2\}},D_{1,\{3\}},D_{1,\{4\}}\}$.
The storage placement at time $t$ follows from Figure~\ref{fig:ex-m-2}a. First, every worker stores the data point needed to be processed. Then, in the excess storage, every worker $w_k$ stores the sub-points labeled by the set $\mathcal{W}$, where $k\in \mathcal{W}$, e.g., $w_1$ stores the sub-point $ \A^t(2,1)= \{D_{2,\{1\}}\}$ from $D_2$.
To summarize, each worker stores the assigned data point of size $d$, and for each one of the remaining $3$ data points, it stores $1$ sub-point of size $d/3$. That is, the storage requirement is given by
$S = 1 + 3\times 1\times 1/3= 2$, which satisfies the storage constraint.

\noindent \underline{Aligned Coded Shuffling} According to the storage placement at time $t$ in Figure~\ref{fig:ex-m-2}a, at time $t+1$ every worker needs $2$ sub-points of the assigned data point, where every needed sub-point is available at exactly $2$ other workers. From an interference perspective, every one of the needed sub-points is an interference to only one worker, e.g., $D_{3,\{4\}}$ needed by $w_2$ at time $t+1$, is available at $w_3$ and $w_4$, and cause interference at $w_1$ (neither needed nor available). Therefore, $w_1$ can face interference from total $2$ sub-points: $D_{3,\{4\}}$ (needed by $w_2$), and $D_{4,\{2\}}$ (needed by $w_3$). By aligning these two sub-points and considering the coded symbol $D_{3,\{4\}}\oplus D_{4,\{2\}}$, we notice the following: 1) This coded symbol is available at the worker $w_4$; 2) It is useful for the two workers $w_2$, and $w_3$ at the same time; and 3) It is the only source of interference for $w_1$. Similarly, we can produce $3$ more aligned symbols to get in total $4$ aligned  coded symbols, of size $d/\binom{3}{1}$ bits each, summarized as follows:
\begin{align}
&D_{3,\{4\}}\oplus D_{4,\{2\}}: \text{Interference at $w_1$, available at $w_4$, and useful for $w_2$, and $w_3$};\nonumber\\
&D_{1,\{3\}}\oplus D_{4,\{1\}}: \text{Interference at $w_2$, available at $w_1$, and useful for $w_3$, and $w_4$};\nonumber\\
&D_{1,\{2\}}\oplus D_{2,\{4\}}: \text{Interference at $w_3$, available at $w_2$, and useful for $w_1$, and $w_4$};\nonumber\\
&D_{2,\{3\}}\oplus D_{3,\{1\}}: \text{Interference at $w_4$, available at $w_3$, and useful for $w_1$, and $w_2$}.
\end{align}
Therefore, these $4$ coded symbols provide every worker with the $2$ needed sub-points. Moreover, it suffices to send only three independent linear combinations of these $4$ coded symbols as shown in Figure~\ref{fig:ex-m-2}b, since every worker already has one of them available locally at its storage. The rate of this transmission is $R=3\times 1/3 =1$, and the pair $(S=2,R=1)$ is achievable which closes the gap in Figure~\ref{fig:ex_gap} for $S=2$.

\noindent \underline{Storage update and  sub-points relabeling}
The storage update at time $t+1$ is done in a way that preserves the structure of the storage at time $t$. 
As shown in Figure~\ref{fig:ex-m-2}b, for every data point $D_i$ (processed by the workers $w_{\delta_{t}(i)}$, and $w_{\delta_{t+1}(i)}$ at epochs $t$, and $t+1$, respectively), the worker $w_{\delta_{t}(i)}$, which already has $D_i$ completely, will only keep the part of $D_i$ stored at time $t$ within the excess storage of the worker $w_{\delta_{t+1}(i)}$, i.e., $\A^{t+1}(\delta_{t+1}(i),\delta_{t}(i))=\A^t(\delta_{t}(i),\delta_{t+1}(i))$. For example, $w_1$ will keep the sub-point $A^{t+1}(4,1) = A^{t}(1,4) = \{D_{1,\{4\}}\}$ of $D_1$ in the excess storage at time $t+1$.

The relabeling process is shown in Figure~\ref{fig:ex-m-2}c  for the sub-points in red dashed frames as follows: for the data point $D_i$, we relabel the sub-points in $\A^{t+1}(\delta_{t+1}(i),\delta_{t}(i))= \{D_{i,\mathcal{W}}\}$, where $\delta_{t+1}(i)\in\mathcal{W}$ by replacing $\delta_{t+1}(i)$ in $\mathcal{W}$ with $\delta_{t}(i)$. For example, the data point $D_1$ is processed by $w_1$, and $w_4$ in the epochs $t$, and $t+1$, respectively. Therefore the following relabeling is done to the sub-points of $D_1$:
\begin{align}
A^{t+1}(4,1) = \{D_{1,\{4\}} \rightarrow D_{1,\{1\}}\}.
\end{align}
which preserves the structure of the storage.

\begin{figure}[t]
  \begin{center}
  \includegraphics[width=0.99\columnwidth]{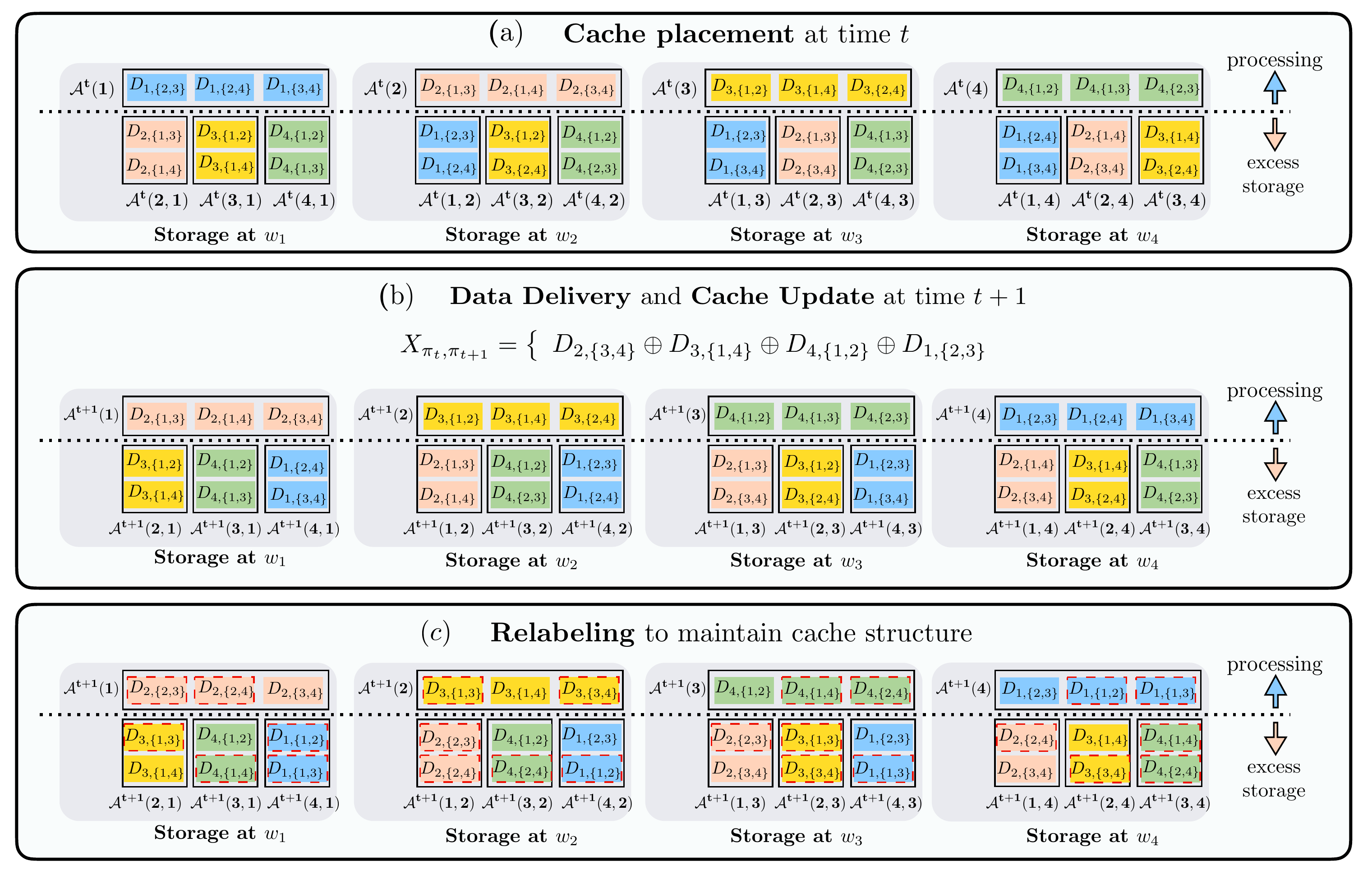}
\caption{An example on closing the gap of $K=4$ workers, $N=4$ data points, and $m=3$ ($S=3$): (a) Structural Invariant Storage placement, (b) Data Delivery and storage update, and (c) Relabeling some sub-points in red dashed frames to maintain the storage structure.
At time $t$, every data point $D_i$ is partitioned into $3$ sub-points each labeled by a unique subset of length $2$ of the set $[1:4]\setminus \delta_t(i)$. 
At time $t+1$, for every data point $D_i$ the sub-point $D_{i,\mathcal{W}}$ where $\delta_{t+1}(i)\in\mathcal{W}$ is relabeled by replacing $\delta_{t+1}(i)$ with $\delta_{t}(i)$.
Above the dotted line is the data point fully stored for processing, and below the dotted line is the excess storage used to store the sub-points containing the worker's index. \label{fig:ex-m-3}}
  \end{center}
\end{figure}

\noindent$\bullet$\hspace{5pt} \textbf{Case} $\mathbf{m=3}$ ($\mathbf{S=3}$):

\noindent \underline{Storage Placement} Every data point at time $t$ is partitioned into $3$ sub-points of size $d/3$ bits each, where every sub-point of the data point $D_i$  is labeled by a unique subset $\mathcal{W} \subseteq [1:4]\setminus \delta_t(i)$, where $\vert \mathcal{W} \vert =2$. For example, the data point $D_1$ at time $t$ is partitioned as $D_1= \{D_{1,\{2,3\}},D_{1,\{2,4\}},D_{1,\{3,4\}}\}$. 
The storage placement at time $t$ follows from Figure~\ref{fig:ex-m-3}a. First, every worker stores the data point needed to be processed. Then, in the excess storage, every worker $w_k$ stores the sub-points labeled by the set $\mathcal{W}$, where $k\in \mathcal{W}$, e.g., $w_1$ stores the two sub-points $\A^t(2,1)= \{D_{2,\{1,3\}},D_{1,\{1,4\}}\}$ from $D_2$.
To summarize, each worker stores the assigned data point of size $d$, and for each one of the remaining $3$ data points, it stores $2$ sub-points of size $d/3$ each. That is, the storage requirement is given by
$S = 1 + 3\times 2\times 1/3= 3$, which satisfies the storage constraint.

\noindent \underline{Aligned Coded Shuffling} According to the storage placement at time $t$ in Figure~\ref{fig:ex-m-3}a, at time $t+1$ every worker needs only one sub-point of the assigned data point, which is available at the $3$ other workers, e.g., $w_1$ needs $D_{2,\{3,4\}}$ which is available at the workers $w_2$, $w_3$, and $w_4$. Therefore, we can send an order $4$ symbol useful for the $4$ workers at the same time as follows:
\begin{align}
X_{\pi_t,\pi_{t+1}} = \left\{ D_{2,\{3,4\}}\oplus D_{3,\{1,4\}}\oplus D_{4,\{1,2\}} \oplus D_{1,\{2,3\}}\right\}.
\end{align}
The rate of this transmission is $R=1\times 1/3 =1/3$, and the pair $(S=3,R=1/3)$ is achievable which closes the gap in Figure~\ref{fig:ex_gap} for $S=3$.

\noindent \underline{Storage update and  sub-points relabeling}
Similar to the case $m=2$, the storage update for the case $m=3$ is shown in Figure~\ref{fig:ex-m-3}b.
For every data point $D_i$, the worker $w_{\delta_{t}(i)}$, which already has $D_i$ completely, will only keep the part of $D_i$ stored at time $t$ within the excess storage of the worker $w_{\delta_{t+1}(i)}$, i.e., $\A^{t+1}(\delta_{t+1}(i),\delta_{t}(i))=\A^t(\delta_{t}(i),\delta_{t+1}(i))$. For example, $w_1$ will keep the sub-points $A^{t+1}(4,1) = A^{t}(1,4) = \{D_{1,\{2,4\}},D_{1,\{3,4\}}\}$ of $D_1$ in the excess storage at time $t+1$.

The relabeling process is shown in Figure~\ref{fig:ex-m-3}c  for the sub-points in red dashed frames in a similar way to the case $m=2$
as follows: for the data point $D_i$, we relabel the sub-points in $\A^{t+1}(\delta_{t+1}(i),\delta_{t}(i))= \{D_{i,\mathcal{W}}\}$, where $\delta_{t+1}(i)\in\mathcal{W}$ by replacing $\delta_{t+1}(i)$ in $\mathcal{W}$ with $\delta_{t}(i)$. For example, the data point $D_1$ is processed by $w_1$, and $w_4$ in the epochs $t$, and $t+1$, respectively. Therefore the following relabeling is done to the sub-points of $D_1$:
\begin{align}
A^{t+1}(4,1)=\{D_{1,\{2,4\}} \rightarrow D_{1,\{1,2\}}, \textbf{ and }  D_{1,\{3,4\}}\rightarrow D_{1,\{1,3\}}\}.
\end{align}
which preserves the structure of the storage.

As a conclusion for the example $K=N=4$, the lower convex envelope of the achievable pairs $(S=m,R=\frac{4-m}{m})$,  for $m\in[1:4]$, is the optimal storage-rate trade-off.
\end{example}

\section{Conclusion}
\label{sec:conclusion}
We considered the worst-case trade-off between the amount of storage and communication overhead for the data shuffling problem. First, we presented an information theoretic formulation of the problem. Following that, we proposed a  novel uncoded-structural invariant storage placement and update strategy for different storage values at the workers. This placement strategy allowed for applying a similar coding scheme to the one in \cite{FundLimitsCaching2015}. Through a novel bounding methodology similar to \cite{Optimality/UncodedCache,Optimality/UncodedCache2}, we derived an information theoretic lower bound on the worst-case communication rate as a function of the storage, which showed that the resulting communication overhead of our scheme is within a maximum multiplicative factor of $\frac{K}{K-1}$, where $K$ is the number of workers.
Furthermore, we presented a new scheme inspired by the idea of interference alignment, which closes the gap and hence achieves the optimal worst-case rate-storage trade-off for $K<5$, and further reduces the maximum multiplicative factor to $\frac{K-\frac{1}{3}}{K-1}$ for $K\geq 5$.

\bibliographystyle{unsrt}
\bibliography{./paper-arxiv.bib}

\begin{appendices}

\section{Upper Bound on $R_{\text{worst-case}}^*$ (Proof of Theorem~\ref{thm1})}
\label{sec:upper-bound}

Following Example~\ref{ex1}, we present our general achievability for any number of workers $K$, any number of data points $N$, and any storage value $S$. Our scheme has two main phases: structural invariant storage placement/update phase; and data delivery phase. This scheme also proves the upper bound on the optimal worst-case rate, i.e.,  $R^{\text{upper}}_{\text{worst-case}}$ stated in Theorem~\ref{thm1}.

\subsection{Structural Invariant Placement}

We first propose a structural invariant placement, which allows applying a similar data delivery scheme to the one proposed in \cite{FundLimitsCaching2015}. The placement procedure involves updating the storage content for each worker after each shuffle in order to maintain the structure of the storage. 
 Since the shuffling process at each time is done randomly, all the data points not being processed by a worker $w_k$ are of equal importance to reduce the communication overhead in the next shuffle. Consequently, the amount of excess storage of size $\left(S-\frac{N}{K}\right)d$ is equally divided among these points, where we assume uncoded storage placement.

We focus on a discrete set of storage values given by $S = \left(1+i(\frac{K-1}{K})\right) \frac{N}{K}$, for $i\in[0:K]$. The values in between can then be achieved by memory sharing as stated in Claim~\ref{cl:1}.
At time $t$, the worker $w_k$ first stores the batch assigned for processing, $\A^t(k)$, in order to satisfy the processing constraint in \eqref{eq:cache-min-desired}, which requires $\frac{N}{K}d$ bits of the available storage. That is if a data point $D \in \A^t(k)$, then $D$ is fully stored in $Z_k^t$.
 The excess storage of size $(S-\frac{N}{K})d=i(\frac{K-1}{K})(\frac{N}{K})d$ is used as follows: every data point $D\in \A$ is divided across the dimension $d$ into $\binom{K}{i}$ non-overlapping parts of size $d/\binom{K}{i}$ bits each, and then each part is labeled by a unique set $\mathcal{W}\subseteq [1:K]$ of size $i$. 
The worker $w_k$ stores the sub-point $D_{\mathcal{W}}$ in the excess storage, where $D\not\in \A^t(k)$, only if $k \in \mathcal{W}$. 
Therefore, the number of sub-points a worker $w_k$ is storing from a point $D \not \in \A^t(k)$; is given by $\binom{K-1}{i-1}$ sub-points. The total number of these points  is $N-\frac{N}{K}=\frac{(K-1)N}{K}$ points. Then, the total size necessary for excess storage is
\begin{align}
\frac{(K-1)N}{K}\times \binom{K-1}{i-1} \times \frac{d}{\binom{K}{i}}=\frac{i(K-1)N}{K^2}d = \left(S-\frac{N}{K}\right)d,
\end{align}
which satisfies the memory constraint.

\subsection{Data Delivery Phase}

Next, we present our proposed delivery scheme to satisfy the  new data assignment characterized by the shuffles $(\pi_t,\pi_{t+1})$. 
According to the adopted placement strategy, whenever a new data point is newly assigned to a worker, it already has $\binom{K-1}{i-1}$ out of the total $\binom{K}{i}$ partitions. Therefore, the number of sub-points still needed for every new assigned data point is $\binom{K}{i}-\binom{K-1}{i-1} = \binom{K-1}{i}$.
Moreover, for the worst-case scenario, each worker is assigned completely new data points, and there are $\frac{N}{K}$ new data points for each worker, i.e., $\A^t(k) \cap \A^{t+1}(k) =\phi$. This gives the total number of data sub-points needed by each worker in the worst case to be $\binom{K-1}{i}\frac{N}{K}$.

According to the placement strategy, every data sub-point $D_{\mathcal{W}}$, is stored at least in $i$ different workers.
 Now, if we pick any set $\mathcal{M}\subseteq [1:K]$ of the workers, where $\vert \mathcal{M} \vert =i+1$, then for each worker $w_k$, where $k \in \mathcal{M}$, and for each point $D$ newly assigned to $w_k$ in the next shuffle, i.e., ${D\not\in \A^t(k)}$, and ${D\in \A^{t+1}(k)}$, there is at least one sub-point needed by $k$ from the remaining workers in the set, labeled as $D_{\mathcal{M}\setminus k}$.
Therefore, we can send $\frac{N}{K}$ order $i+1$ coded symbols in the form $\underset{k\in\mathcal{M}}{\oplus} A^{t+1}_{\mathcal{M}\setminus k}(k)$, of size ${d}/{\binom{K}{i}}$ each, useful for all the $i+1$ workers in $\mathcal{M}$ in the same time.

Considering all the possible sets $\mathcal{M}$ of size $i+1$, this process is repeated $\binom{K}{i+1}$ number of times, which gives $\binom{K}{i+1}\frac{N}{K}$ coded symbols given by
\begin{align}
X_{\pi_t,\pi_{t+1}}= \left\{\underset{\substack{\mathcal{M}\subseteq [1:K]\\ \vert\mathcal{M}\vert =i+1}}{\cup}\underset{k\in\mathcal{M}}{\oplus} A^{t+1}_{\mathcal{M}\setminus k}(k)\right\}.\label{eq:scheme-trans}
\end{align}
The corresponding total worst-case number of bits sent over the shared link is given by
\begin{align}
R_{\text{worst-case}}d =\binom{K}{i+1} \times \frac{N}{K} \times \frac{d}{\binom{K}{i}} = \frac{N(K-i)}{K(i+1)}d.
\end{align}
It is important to notice that the total number of times $w_k$ becomes a member of the set $\mathcal{M}$ is $\binom{K-1}{i}$.
Therefore, by sending the coded symbols in \eqref{eq:scheme-trans}, every worker gets $\binom{K-1}{i}\frac{N}{K}$ sub-points in total, which are enough to recover $\frac{N}{K}$ data points in the worst-case scenario as previously discussed.

Using the memory sharing concept in Claim~\ref{cl:1}, we can achieve the lower convex envelope of the following $K+1$ points:
\begin{align}
\left(S=\left(1+i\frac{K-1}{K}\right)\frac{N}{K},\:R_{\text{worst-case}}^{\text{upper}}= \frac{N(K-i)}{K(i+1)}\right), \quad \forall i\in[0:K],
\end{align}
which completes the proof of Theorem~\ref{thm1}.

\subsection{Storage Update Procedure}

In order to maintain the structure of the storage after the next shuffle at time $t+1$, the storage update procedure takes place at worker $w_k$ for every point $D\in \A$ according to the following cases:

\noindent$\bullet$\hspace{5pt} \underline{${D\in \A^t(k)}$, and ${D\in \A^{t+1}(k)}$:} In this case $D$ remains stored completely in $Z^{t+1}_k$.

\noindent$\bullet$\hspace{5pt} \underline{${D\not\in \A_k^t}$, and ${D\in \A_k^{t+1}}$:} After the data delivery, worker $w_k$ stores $D$ completely in $Z^{t+1}_k$.

\noindent$\bullet$\hspace{5pt} \underline{${D\in \A^t(k)}$, and ${D\not\in \A^{t+1}(k)}$:} Out of the point $D$  previously stored completely in $Z^{t}_k$, worker $w_k$ chooses to stores in $Z^{t+1}_k$ the sub-points $D_{\mathcal{W}}$ where $k \in {\mathcal{W}}$.

\noindent$\bullet$\hspace{5pt} \underline{${D\not\in \A^t(k)}$, and ${D\not\in \A^{t+1}(k)}$:} Nothing changes about the storage of $D$ in the excess storage of $Z^{t+1}_k$, and $w_k$ keeps the same sub-points of $D$ previously stored in $Z^{t}_k$, i.e., $D_{\mathcal{W}}$ where $k\in {\mathcal{W}}$.

\section{Lower Bound on $R_{\text{worst-case}}^*$ (Proof of Theorem~\ref{thm2})}
\label{sec:lower-bound}

In this section, we present an information theoretic lower bound on the worst-case communication rate.
Following Remark~\ref{re:basic-idea-converse}, we start by considering the following shuffle $(\pi_t,\pi_{t+1})$ at time $t+1$: for a permutation of the worker indexes $\sigma: (1,2,\ldots,K) \rightarrow (\sigma_1,\sigma_2,\ldots,\sigma_K)$, the worker $w_{\sigma_{k}}$ at time $t+1$ is assigned the data batch that was assigned to the worker $w_{\sigma_{k-1}}$ at time $t$, i.e., $\A^{t+1}({\sigma_{k}}) =\A^{t}({\sigma_{k-1}})$, which also gives the following condition using (\ref{eq:decoding-const}):
\begin{align}
\label{eq:decoding_const2}
H(\A^{t+1}({\sigma_{k}})| Z^{t}_{\sigma_{k}},X_{\pi_t,\pi_{t+1}}) = H(\A^{t}({\sigma_{k-1}})| Z^{t}_{\sigma_{k}},X_{\pi_t,\pi_{t+1}})=0.
\end{align}
Next, we prove that $H(\A|\Z^t_{\sigma_{[2:K]}},X_{\pi_t,\pi_{t+1}}) =0$ using (\ref{eq:cache-min-desired}), and (\ref{eq:decoding_const2}) as follows:
\begin{align}
H(\A|\Z^t_{\sigma_{[2:K]}},X_{\pi_t,\pi_{t+1}}) &= H(\A^t({[1:K]})|\Z^t_{\sigma_{[2:K]}},X_{\pi_t,\pi_{t+1}})\leq \sum_{j=1}^K H(\A^t({\sigma_j})|\Z^t_{\sigma_{[2:K]}},X_{\pi_t,\pi_{t+1}})\nonumber\\
& \leq \sum_{j=2}^K H(\A^t({\sigma_j})|Z^t_{\sigma_j}) +H(\A^t({\sigma_1})|Z^t_{\sigma_2},X_{\pi_t,\pi_{t+1}}) =0.\label{eq:decoding_const3}
\end{align}
Using (\ref{eq:decoding_const2}), and (\ref{eq:decoding_const3}), we obtain the following bound:
\begin{align}
Nd &=H(\A) = I\left(\A;\Z^t_{\sigma_{[2:K]}},X_{\pi_t,\pi_{t+1}}\right) +H\left(\A|\Z^t_{\sigma_{[2:K]}},X_{\pi_t,\pi_{t+1}}\right)\nonumber\\
& \overset{(a)}{\leq}H\left(\Z^t_{\sigma_{[2:K]}},X_{\pi_t,\pi_{t+1}}\right)-H\left(\Z^t_{\sigma_{[2:K]}},X_{\pi_t,\pi_{t+1}}|\A\right) \nonumber\\
&\overset{(b)}{=} H\left(X_{\pi_t,\pi_{t+1}},Z^t_{\sigma_K}\right) + H\left(\Z^t_{\sigma_{[2:K-1]}]}|X_{\pi_t,\pi_{t+1}},Z^t_{\sigma_K}\right)\nonumber\\
&\leq H\left(X_{\pi_t,\pi_{t+1}}\right)+H\left(Z^t_{\sigma_K}\right) + \sum_{i=2}^{K-1}H\left(Z^t_{\sigma_i}|\Z^t_{\sigma_{[i+1:K]}},X_{\pi_t,\pi_{t+1}}\right)\nonumber\\
& \overset{(c)}{\leq} R_{\pi_t,\pi_{t+1}}^*d +H\left(\A^t({\sigma_K})\right)+H\left(\A^t({\sigma_{[1:K-1]}},\sigma_K)\right)\nonumber\\
&\hspace{40pt}+ \sum_{i=2}^{K-1}H\left(Z^t_{\sigma_i}|\Z^t_{\sigma_{[i+1:K]}},X_{\pi_t,\pi_{t+1}},\A^t({\sigma_{[i:K]}})\right)\nonumber
\end{align}
\begin{align}
& \overset{(d)}{=} R_{\pi_t,\pi_{t+1}}^*d + \frac{N}{K}d+ H\left(\A^t({\sigma_{[1:K-1]}},\sigma_K)\right) \nonumber\\
&\hspace{40pt}+ \sum_{i=2}^{K-1}H\left(\A^t({\sigma_{[1:i-1]}},\sigma_i)|\Z^t_{\sigma_{[i+1:K]}},X_{\pi_t,\pi_{t+1}},\A^t({\sigma_{[i:K]}})\right)\nonumber\\
& \overset{(e)}{\leq} R_{\pi_t,\pi_{t+1}}^*d + \frac{N}{K}d+H\left(\A^t({\sigma_{[1:K-1]}},\sigma_K)\right)+ \sum_{i=2}^{K-1}H\left(\A^t({\sigma_{[1:i-1]}},\sigma_i)|\Z^t_{\sigma_{[i+1:K]}}\right)
\nonumber\\
& = R_{\pi_t,\pi_{t+1}}^*d + \frac{N}{K}d+\sum_{i=2}^{K}\sum_{j=1}^{i-1}H\left(\A^t(\sigma_j,\sigma_i)|\Z^t_{[\sigma_{i+1}:\sigma_K]}\right)\nonumber\\
& \overset{(f)}{=} R_{\pi_t,\pi_{t+1}}^*d + \frac{N}{K}d+  \sum_{i=2}^{K}\sum_{j=1}^{i-1}H\left(\A^t({\sigma_j},\sigma_i)|\A^t({\sigma_j},\sigma_{[i+1:K]})\right)\nonumber\\
& = R_{\pi_t,\pi_{t+1}}^*d + \frac{N}{K}d + \sum_{j=1}^{K-1}\sum_{i=j+1}^{K}H\left(\A^t({\sigma_j},\sigma_i)|\A^t({\sigma_j},\sigma_{[i+1:K]})\right)\nonumber\\
& \overset{(g)}{=} R_{\pi_t,\pi_{t+1}}^*d + \frac{N}{K}d + \sum_{j=1}^{K-1}H\left(\A^t({\sigma_j},\sigma_{[j+1:K]})\right),\label{eq:bound1}
\end{align}
where $(a)$ follows from \eqref{eq:decoding_const3}, $(b)$ follows from \eqref{eq:cache-content}, and \eqref{eq:transmit-load}, where $\Z^t_{\sigma_{[2:K]}}$, and $X_{\pi_t,\pi_{t+1}}$ are deterministic functions of the data-set  $\A$, $(c)$ follows from (\ref{eq:cache-min-desired}), \eqref{eq:decoding_const2}, and the constraint on storage in \eqref{eq:cache-content2}, $(d)$ also follows  from \eqref{eq:cache-content2} where after knowing $\A^t({\sigma_{[i:K]}})$, the only parts left in $Z_{\sigma_i}^t$ are $\A^t({\sigma_{[1:i-1]}},\sigma_i)$, $(e)$ because conditioning reduces entropy, $(f)$ follows since $\A^t({\sigma_j},\sigma_i)$ only depends on the parts of the batch $\A^t({\sigma_j})$ stored at $\Z^t_{[\sigma_{i+1}:\sigma_K]}$, and finally $(g)$ follows from the chain rule of entropy.
 From the definition in \eqref{eq:def_A_k_W}, we can write $\A^t({\sigma_j},\sigma_{[j+1:K]})$ as
\begin{align}
\A^t({\sigma_j},\sigma_{[j+1:K]}) = \underset{\substack{\mathcal{S}\subseteq \sigma_{[j+1:K]}: \: \mathcal{S}\neq \phi}}{\cup} \hspace{3pt}
\underset{\mathcal{W} \subseteq [1:K]\setminus \sigma_j: \: \mathcal{S} \in \mathcal{W}}{\cup} \A^t_{\mathcal{W}}(\sigma_j)
\end{align}
Therefore, we can upper bound the entropy $H\left(\A^t({\sigma_j},\sigma_{[j+1:K]})\right)$ as
\begin{align}
H\left(\A^t({\sigma_j},\sigma_{[j+1:K]})\right) &\leq \sum_{\substack{\mathcal{S}\subseteq \sigma_{[j+1:K]}: \: \mathcal{S}\neq \phi}} \hspace{3pt}\sum_{\mathcal{W} \subseteq [1:K]\setminus\sigma_j: \: \mathcal{S} \in \mathcal{W}} \vert \A^t_{\mathcal{W}}(\sigma_j)\vert \hspace{2pt} d\nonumber\\
&=\sum_{\mathcal{W} \subseteq [1:K]\setminus \sigma_j} \vert \A^t_{\mathcal{W}}(\sigma_j)\vert \hspace{2pt} d \hspace{5pt}- \sum_{\mathcal{W} \subseteq \sigma_{[1:j-1]}} \vert \A^t_{\mathcal{W}}(\sigma_j)\vert \hspace{2pt} d,\label{eq:bound2}
\end{align}
where $\vert \A^t_{\mathcal{W}}(j)\vert$ is the size of the sub-batch $\A^t_{\mathcal{W}}(j)$ normalized by $d$.
Therefore, by applying (\ref{eq:bound2}) in (\ref{eq:bound1}), we get a bound over $R_{\pi_t,\pi_{t+1}}^*$, which is also a lower bound over $R_{\text{worst-case}}^*$ following Remark~\ref{re:basic-idea-converse}, as follows:
\begin{align}
R^*_{\text{worst-case}}\geq R_{\pi_t,\pi_{t+1}}&\geq N-\frac{N}{K} -\sum_{j=1}^{K-1}\hspace{3pt}\left[ \sum_{\mathcal{W} \subseteq [1:K]\setminus \sigma_j } \vert \A^t_{\mathcal{W}}(\sigma_j)\vert \hspace{2pt} \hspace{5pt}- \sum_{\mathcal{W} \subseteq \sigma_{[1:j-1]}} \vert \A^t_{\mathcal{W}}(\sigma_j)\vert \right] \nonumber\\ 
&=N-\frac{N}{K} -\sum_{\ell=0}^{K-1}\sum_{j=1}^{K-1}\hspace{3pt}\left[ \sum_{\substack{\mathcal{W} \subseteq [1:K]\setminus\sigma_j\\ \vert\mathcal{W}\vert =\ell}} \vert \A^t_{\mathcal{W}}(\sigma_j)\vert \hspace{5pt}- \sum_{\substack{\mathcal{W} \subseteq \sigma_{[1:j-1]}\\ \vert\mathcal{W}\vert =\ell}} \vert \A^t_{\mathcal{W}}(\sigma_j)\vert \right].\label{eq:bound3}
\end{align}

For $K!$ possible permutations $\sigma$ of the ordered set $(1,2,\ldots,K)$, we get $K!$ different bounds over $R^*_{\text{worst-case}}$ from (\ref{eq:bound3}). Summing up over all the possible $K!$ permutations $\sigma$, we get
\begin{align}
R^*_{\text{worst-case}}\geq N-\frac{N}{K} -\frac{1}{K!}
\sum_{\ell=0}^{K-1}\sum_{j=1}^{K-1} \sum_{\sigma\in [K!]}\hspace{3pt}\left[  \sum_{\substack{\mathcal{W} \subseteq [1:K]\setminus\sigma_j\\ \vert\mathcal{W}\vert =\ell}} \vert \A^t_{\mathcal{W}}(\sigma_j)\vert \hspace{5pt}- \sum_{\substack{\mathcal{W} \subseteq \sigma_{[1:j-1]}\\ \vert\mathcal{W}\vert =\ell}} \vert \A^t_{\mathcal{W}}(\sigma_j)\vert \right],
\label{eq:bound4}
\end{align}
where $[K!]$ is defined as the set of all possible permutations of the ordered set $(1,2,\ldots,K)$, which contains $K!$ permutations. Due to symmetry, for each value of a $(\ell,j)$ pair in the outer summation in \eqref{eq:bound4}, where $\ell\in [0:K-1]$ and $j\in [1:K-1]$, the coefficients of each $\vert \A^t_{\mathcal{W}}(k) \vert$ in the inner summation for $k\in [1:K]$  and $\vert\mathcal{W}\vert =\ell$ are equal.
Therefore, we can write the inner summation in (\ref{eq:bound4}) in the following form:
\begin{align}
\sum_{\sigma\in [K!]}\left[  \sum_{\substack{\mathcal{W} \subseteq [1:K]\setminus \sigma_j\\ \vert\mathcal{W}\vert =\ell}} \vert \A^t_{\mathcal{W}}(\sigma_j)\vert - \sum_{\substack{\mathcal{W} \subseteq \sigma_{[1:j-1]}\\ \vert\mathcal{W}\vert =\ell}} \vert \A^t_{\mathcal{W}}(\sigma_j)\vert \right]&= \left(c_1^{j,\ell} -c_2^{j,\ell}\right)\sum_{k=1}^K \sum_{\mathcal{W}\subseteq [1:K]:\: \vert \mathcal{W}\vert = \ell} \vert \A^t_{\mathcal{W}}(k)\vert\nonumber\\
&= \left(c_1^{j,\ell} -c_2^{j,\ell}\right)x_{\ell}, \label{eq:app1}
\end{align}
where $c^{j,\ell}_1$, and $c^{j,\ell}_2$ are the two coefficients of $x_{\ell}$ coming from the two inner summations in the LHS of \eqref{eq:app1}.
From (\ref{eq:app1}), finding $c^{j,\ell}_1$, and $c^{j,\ell}_2$ is the same as finding the coefficients of one realization of $k$, and $\mathcal{W}$ on the right side of the equation, and we consider for instance $\A^t_{[2:\ell+1]}(1)$. In the first sum, we get $c^{j,\ell}_1$ by counting the number of permutations where $\sigma_j =1$, which is given by
\begin{align}
c_1^{j,\ell} = (K-1)!.
\end{align}
In the second sum, we get $c^{j,\ell}_2$ by counting the number of permutations such that $\sigma_j =1$, and $\sigma_{j+1},\ldots, \sigma_K \in[\ell+2:K]$, which is given by
\begin{align}
c^{j,\ell}_2 = \frac{(K-\ell-1)!}{(j-\ell-1)!} (j-1)! = \frac{\binom{j-1}{\ell}}{\binom{K-1}{\ell}} (K-1)!.
\end{align}
Therefore, we can write the summation  in \eqref{eq:app1} in the following form:
\begin{align}
\sum_{\sigma\in [K!]}\left[  \sum_{\substack{\mathcal{W} \subseteq [1:K]\setminus \sigma_j\\ \vert\mathcal{W}\vert =\ell}} \vert \A^t_{\mathcal{W}}(\sigma_j)\vert - \sum_{\substack{\mathcal{W} \subseteq \sigma_{[1:j-1]}\\ \vert\mathcal{W}\vert =\ell}} \vert \A^t_{\mathcal{W}}(\sigma_j)\vert \right]= \left((K-1)! -\frac{\binom{j-1}{\ell}}{\binom{K-1}{\ell}} (K-1)!\right)x_{\ell}. \label{eq:app2}
\end{align}

Now, we use \eqref{eq:app2} in \eqref{eq:bound4} to obtain the following bound:
\begin{align}
R^*_{\text{worst-case}}&\geq  N-\frac{N}{K} -\frac{1}{K!}
\sum_{\ell=0}^{K-1}\sum_{j=1}^{K-1} \left[ (K-1)! - 
\frac{\binom{j-1}{\ell}}{\binom{K-1}{\ell}}(K-1)!
 \right] x_{\ell}\nonumber\\
& = N-\frac{N}{K}-\frac{1}{K!} \sum_{\ell=0}^{K-1} \left[(K-1) (K-1)! - \frac{\binom{K-1}{\ell+1}}{\binom{K-1}{\ell}}(K-1)! \right]x_{\ell}\nonumber\\
& =  N-\frac{N}{K} -\frac{1}{K} \sum_{\ell=0}^{K-1} \left[(K-1)- \frac{K-\ell-1}{t+1} \right]x_{\ell}\nonumber\\
& \overset{(a)}{=}  \sum_{\ell=0}^{K-1} x_{\ell}-\frac{N}{K} - \sum_{\ell=0}^{K-1}  \frac{\ell}{\ell+1} x_{\ell} =\sum_{\ell=0}^{K-1}  \frac{1}{\ell+1} x_{\ell} -\frac{N}{K},\label{eq:objective}
\end{align}
where $(a)$ follows from the data size constraint in \eqref{eq:size-const}.
Next, we  obtain $K-1$ different lower bounds on the optimal worst-case transmission rate $R^*_{\text{wc}}$, by eliminating the pairs $(x_{j-1},x_{j})$, for each ${j \in \left[1:K-1\right]}$, in the equation \eqref{eq:objective} using the equations \eqref{eq:size-const} and \eqref{eq:constraint_cache}. 
We use \eqref{eq:size-const} to write $x_{j-1}$ as follows:
\begin{align}
\label{eq:x_j-1}
&x_{j-1}=N-\sum_{\ell\in\left[0:K\right]\setminus j-1}x_{\ell}.
\end{align}
We first apply \eqref{eq:x_j-1} in \eqref{eq:objective} to obtain
\begin{align}
R^*_{\text{worst-case}}&\geq \sum_{\ell\in[0:K-1]\setminus j-1} \frac{1}{\ell+1} x_{\ell} +\frac{1}{j} \left( N- \sum_{\ell\in[0:K-1]\setminus j-1}x_{\ell}\right) -\frac{N}{K}\nonumber\\
& = \frac{N(K-j)}{Kj}-\sum_{\ell\in[0:K-1]\setminus j-1}\frac{\ell-j+1}{j(\ell+1)}x_{\ell}.
\label{eq:rate_1}
\end{align}
We next apply \eqref{eq:x_j-1} in the excess storage constraint of \eqref{eq:constraint_cache} to obtain
\begin{align}
&\sum_{\ell\in\left[0:K-1\right]\setminus j-1}\ell x_{\ell}+(j-1)\left(N-\sum_{\ell\in[0:K]\setminus j-1}x_{\ell}\right)\leq K\left(S-\frac{N}{K}\right), \nonumber \\
&\sum_{\ell\in\left[0:K-1\right]\setminus j-1}\left(\ell-j+1\right)x_
{\ell} \leq K\left(S-j\frac{N}{K}\right). \label{eq:cache1} 
\end{align}
Now, we need to eliminate $x_{j}$ from \eqref{eq:rate_1}. We use \eqref{eq:cache1} to bound $x_{j}$ as
\begin{align}
x_{j}\leq K\left(S-j\frac{N}{K}\right)-\sum_{\ell\in\left[0:K-1\right]\setminus \{j-1,j\}}\left(\ell-j+1\right)x_{\ell}. \label{eq:cache2}
\end{align}
Then, we use this bound in \eqref{eq:rate_1} as follows:
\begin{align}
&R^*_{\text{worst-case}}\nonumber\\
&\geq \frac{N(K-j)}{Kj}-\sum_{\ell\in[0:K-1]\setminus \{j-1,j\}}\frac{\ell-j+1}{j(\ell+1)}x_{\ell} -\frac{1}{j(j+1)} x_{j} \nonumber \\ 
&\overset{(a)}{\geq}
\frac{N(K-j)}{Kj}-\sum_{\ell\in[0:K-1]\setminus \{j-1,j\}}\frac{\ell-j+1}{j(\ell+1)}x_{\ell} -\frac{K\left(S-j\frac{N}{K}\right)}{j(j+1)} 
+\sum_{\ell\in\left[0:K-1\right]\setminus \{j-1,j\}}
\frac{\left(\ell-j+1\right)}{j(j+1)} x_{\ell}
\nonumber\\
&=\frac{N(K-j)}{Kj} -\frac{K\left(S-j\frac{N}{K}\right)}{j(j+1)}   +\sum_{\ell\in\left[0:K-1\right]\setminus \{j-1,j\}}\lambda_{\ell} x_{\ell}\nonumber\\
&\overset{(b)}{\geq} \frac{N(K-j)}{Kj}-\frac{K\left(S-j\frac{N}{K}\right)}{j(j+1)}, \label{eq:lower_in1}
\end{align}
where $(a)$ follows from \eqref{eq:cache2} where the coefficient of $x_{j}$ in the above equation is negative for all ${j\in\left[1:K-1\right]}$, and $(b)$ since the coefficients, $\lambda_{\ell}$, of $x_{\ell}>0$ are positive for ${\ell\in\left[0:K-1\right]\setminus \{j-1,j\}}$, which can be shown in the following:
\begin{align}
\lambda_{\ell}&=\frac{\ell-j+1}{j(j+1)}-\frac{\ell-j+1}{j(\ell+1)} = \frac{(\ell-j)(\ell-j+1)}{j(j+1)(\ell+1)},
\end{align}
where ${j,j+1,\ell+1>0}$ for $\ell,j\geq 0$, then we only need to show that $(\ell-j)(\ell-j+1)>0$ for ${\ell\in\left[0:K-1\right]\setminus \{j-1,j\}}$. This can be easily checked by assuming $y=\ell-j$, then $y(y+1)$ is only negative in the range $-1<y<0$, or $j-1<\ell<j$, which is not in the range of $\ell$ in the above summation.

The lower bound in \eqref{eq:lower_in1} is a linear function of $S$ for a fixed value of $j\in[1:K-1]$ passing through the points $\left(S_1=j\frac{N}{K},R_1= \frac{N(K-j)}{Kj}\right)$, and $\left(S_2=(j+1)\frac{N}{K},R_2= \frac{N(K-j-1)}{K(j+1)}\right)$. We obtain $K-1$ such lower bounds for every $j\in[1:K-1]$, which eventually give the lower bound over $R^*_{\text{worst-case}}$ as the lower convex envelope of the following $K$ points:
\begin{align}
\left(S=m\frac{N}{K},\:R_{\text{worst-case}}^{\text{lower}}= \frac{N(K-m)}{Km}\right), \quad \forall m\in[1:K],
\end{align}
which completes the proof of Theorem~\ref{thm2}.

\section{Maximum Gap Analysis (Proof of Theorem~\ref{thm3})}
\label{sec:gap-analysis}
To characterize the maximum gap between the obtained bounds over $R_{\text{worst-case}}^*$, we first express the storage $S$ as multiples of $\frac{N}{K}$, i.e., $S = m \frac{N}{K}$, for $1\leq m\leq K$.
From Theorem~\ref{thm1} for $\left(1+i\frac{K-1}{K}\right)\leq m\leq \left(1+(i+1)\frac{K-1}{K}\right)$, and $i\in[0:K-1]$, we can achieve the line joining the two points $\big(m=\left(1+i\frac{K-1}{K}\right)$, $R= \frac{N(K-i)}{K(i+1)}\big)$, and $\big(m=\left(1+(i+1)\frac{K-1}{K}\right), R= \frac{N(K-i-1)}{K(i+2)}\big)$, which gives the following upper bounds over $R_{\text{worst-case}}^*$ as
\begin{align}
&\frac{R_{\text{worst-case}}^{\text{upper}}- \frac{N(K-i)}{K(i+1)}}{m-\left(1+i\frac{K-1}{K}\right)} = \frac{\frac{N(K-i-1)}{K(i+2)}- \frac{N(K-i)}{K(i+1)}}{\left(1+(i+1)\frac{K-1}{K}\right)-\left(1+i\frac{K-1}{K}\right)} = -\frac{N(K+1)}{(K-1)(i+1)(i+2)},\nonumber\\
&R_{\text{worst-case}}^{\text{upper}} = \frac{N(K-i)}{K(i+1)} -\frac{N(K+1)}{(K-1)(i+1)(i+2)}\left( m-1-i\frac{K-1}{K}\right),\label{eq:upper-bound-gap}
\end{align}
for $\left(1+i\frac{K-1}{K}\right) \leq m\leq \left(1+(i+1)\frac{K-1}{K}\right)$, and $i\in[0:K-1]$.
Also, from \eqref{eq:lower_in1} we have the lower bounds over $R_{\text{worst-case}}^* $ as
\begin{align}
R_{\text{worst-case}}^{\text{lower}} = \frac{N(K-j)}{Kj}-\frac{N\left(m-j\right)}{j(j+1)},\label{eq:lower-bound-gap}
\end{align}
for $j \leq m \leq j+1$, and $j\in[1:K-1]$.

Due to the properties of the piece-wise linear functions, we obtain the maximum gap at one of the following $2K-1$ values of $m$: $m=j$ , for $j \in[1:K-1]$, or $m=1+i\frac{K-1}{K} $, for $i\in[1:K]$.

\subsection{Gap Analysis for $m =1+i\frac{K-1}{K}$, and $i\in[1:K]$}
We first notice that when $i\in[1:K]$, then $i \leq m\leq i+1$. Therefore, the lower bound  $R_{\text{worst-case}}^{\text{lower}}$ at $m=1+i\frac{K-1}{K}$ follows from \eqref{eq:lower-bound-gap} where $j=i$:
\begin{align}
R^{\text{lower}}_{\text{worst-case}}\left(m=1+i\frac{K-1}{K}\right)& = \frac{N(K-i)}{Ki}-\frac{N\left(1+i\frac{K-1}{K}-i\right)}{i(i+1)}\nonumber\\
&=\frac{N(K-i)}{Ki}-\frac{N\left(K-i\right)}{Ki(i+1)} = \frac{N(K-i)}{K(i+1)},\label{eq:lower-bound-gap2}
\end{align}
which matches  the upper bound in \eqref{eq:upper-bound-gap}, when $m =1+i\frac{K-1}{K}$. Therefore, the proposed achievable scheme is optimal for $m =1+i\frac{K-1}{K}$, where $i \in [1:K]$.

\subsection{Gap Analysis for $m =j$, and $j\in[1:K-1]$}
We first notice that when $m=j$, then $\left(1+(j-1)\frac{K-1}{K}\right) \leq m\leq\left(1+j\frac{K-1}{K}\right)$ for $j \in [1:K-1]$. Therefore, the upper bound  $R_{\text{worst-case}}^{\text{upper}}$ at $m=j$ follows from \eqref{eq:upper-bound-gap} where $i=j-1$:
\begin{align}
R_{\text{worst-case}}^{\text{upper}}(m=j) &= \frac{N(K-j+1)}{Kj} -\frac{N(K+1)}{j(K-1)(j+1)}\left( j-1-(j-1)\frac{K-1}{K} \right)\nonumber\\
&= \frac{N(K-j+1)}{Kj} -\frac{N(K+1)(j-1)}{jK(K-1)(j+1)}\nonumber\\
&=\frac{N(K-j)}{Kj} +\frac{N}{Kj}\left(1- \frac{(K+1)(j-1)}{(K-1)(j+1)}\right) ,\label{eq:upper-bound-m1}
\end{align}
whereas the lower bound on $R_{\text{worst-case}}^*(m=j)$ follows from \eqref{eq:lower-bound-gap} directly as follows:
\begin{align}
R_{\text{worst-case}}^{\text{lower}}(m=j)& = \frac{N(K-j)}{Kj}.\label{eq:lower-bound-m1}
\end{align}
Hence, the ratio between the bounds follows by dividing \eqref{eq:upper-bound-m1} by \eqref{eq:lower-bound-m1} as
\begin{align}
\label{eq:gap-ratio}
\frac{R_{\text{worst-case}}^{\text{upper}}}{R_{\text{worst-case}}^{\text{lower}}}=1 + \frac{1}{K-j}\left(1 -\frac{(K+1)(j-1)}{(K-1)(j+1)}\right) =
1+ \frac{2}{(K-1)(j+1)}, \quad j\in[1:K-1]. 
\end{align}
We notice that ratio in \eqref{eq:gap-ratio} is a decreasing function in $j$. Therefore, we obtain the maximum gap with the smallest value of $j$, i.e., $j=1$, which is the no excess storage case $S =\frac{N}{K}$.
Applying $j=1$ in \eqref{eq:gap-ratio}, we obtain the maximum gap ratio as follows:
\begin{align}
\frac{R_{\text{worst-case}}^{\text{upper}}}{R_{\text{worst-case}}^{\text{lower}}}= 1+ \frac{1}{(K-1)} = \frac{K}{K-1},
\end{align}
 which completes the proof of Theorem~\ref{thm3}.

\section{Closing the Gap (Proof of Theorem~\ref{thm4})}
\label{sec:close-gap}

Based on Example~\ref{ex3}, we introduce the general achievability to close the gap for some storage values. In particular, we consider the storage values $S=m\frac{N}{K}$, for $m\in\{1,K-2,K-1\}$, any number of workers $K$, and any number of data points $N$. We also consider \textit{a variation of the structural invariant storage placement}. Every data point $D_i$ for $i\in[1:N]$ is now partitioned into $\binom{K-1}{m-1}$ non-overlapped sub-points. As suggested in the Example~\ref{ex3}, the labeling for the data sub-points is changing over time as follows: At the time epoch $t$, the data sub-points of the data point $D_i$ are labeled by unique subsets $\mathcal{W}_t\subseteq[1:K]\setminus\delta_t(i)$, where $\delta_t(i)$ is the index of the worker assigned to the data point $D_i$ at time $t$. Every worker stores the assigned data points as well as the data sub-points having the worker's index in their labels.
Therefore, any partition of a data point is stored at total number of $m$ workers; $m-1$ workers are storing it as excess storage, and $1$ worker is assigned the whole corresponding data point for processing.

 For invariant structure placement, the change in the labels at time $t+1$ is required only for the data sub-points $D_{i,\mathcal{W}_t}$, where $\delta_{t+1}(i) \in \mathcal{W}_t$, by replacing $\delta_{t+1}(i)$ with $\delta_{t}(i)$ in the label $\mathcal{W}_t$ to obtain the newly labeled sub-points $D_{i,\mathcal{W}_{t+1}}$ where  $\delta_{t}(i)\in \mathcal{W}_{t+1}$.
Therefore, these newly labeled sub-points are required now to be stored  in the excess storage of the worker $w_{\delta_{t}(i)}$, which already has the data point $D_{i}$ fully available at its cache at time $t$. Then, there is no need to deliver these sub-points, and the storage structure can be preserved.

The number of data sub-points of the point $D_i$ needed to be stored at worker $w_k$ at time $t$, where $\delta_{t}(i)\neq k$ ($D_i\not\in \A^t(k)$) is $\binom{K-2}{m-2}$ of size $d/\binom{K-1}{m-1}$ bits each. In total, we have $(K-1)\frac{N}{K}$ such data points where $\delta_{t}(i)\neq k$ for the worker $w_k$. Therefore, the worker $w_k$ needs to store in the excess storage data of total size 
\begin{align}
(K-1)\frac{N}{K} \times \binom{K-2}{m-2} \times \frac{d}{\binom{K-1}{m-1}} = (m-1)\frac{N}{K}d = \left(S-\frac{N}{K}\right)d,
\end{align}
which satisfies the memory constraint.

Before we proceed to the delivery mechanism we define $\A^{t,t+1}(i;j)= \A^t(i) \cap \A^{t+1}(j)$ as the part of data assigned to $w_j$ at time $t+1$ which was also assigned to $w_i$ at time $t$. 
Furthermore, we define $S_{i,j}^{t,t+1} = \vert \A^{t,t+1}(i;j) \vert $ as the number of such data points. 
Therefore, the data batches $\A^{t}(i)$ and $\A^{t+1}(i)$ can then be written as
\begin{align}
\A^{t}(i) = \cup_{j=1}^K \A^{t,t+1}(i;j), \qquad \A^{t+1}(i) = \cup_{j=1}^K \A^{t,t+1}(j;i).
\end{align}
Since we have the size of the data batches is fixed as $\vert \A^{t}(i)\vert = \vert \A^{t+1}(i)\vert = \frac{N}{K}$, we obtain the following property:
\begin{align}
\label{eq:data-conservation}
 \sum_{j=1}^K S^{t,t+1}_{i,j}= \sum_{j=1}^K S^{t,t+1}_{j,i} = \frac{N}{K}.
\end{align}

\begin{remark}[Data-flow Conservation Property] \normalfont
We next state an important property satisfied by any shuffle, namely the data-flow conservation property:
\begin{equation}
\label{eq:dataflow-conservation}
\underset{j\in [1:K]\setminus i}{\sum}S^{t,t+1}_{i,j}= \underset{j\in [1:K]\setminus i}{\sum}S^{t,t+1}_{j,i}.
\end{equation}
The proof of this property follows directly by subtracting $S^{t,t+1}_{i,i}$ from the two sides of \eqref{eq:data-conservation}, and has the following interesting interpretation: the total number of new data points that need to be delivered to worker $w_i$ (and are present elsewhere), i.e., the RHS of \eqref{eq:dataflow-conservation}, is exactly equal to the total number of data points that worker $w_i$ has that are desired by the other workers, which is the LHS of \eqref{eq:dataflow-conservation}. 
\end{remark}

The rate $R_{\pi_t,\pi_{t+1}}$ is characterized by $S_{i,j}^{t,t+1}$ for $i,j\in[1:K]$. These shuffling parameters can be held in the matrix $S^{t,t+1} =[S_{i,j}^{t,t+1}]_{i,j}$, which can be named as the \textit{shuffling matrix}.
Moreover, according to the property in \eqref{eq:data-conservation} the shuffling matrix $S^{t,t+1}$ is a $K\times K$ square matrix with the row sum equals the column sum equals $\frac{N}{K}$. In the following discussion, we drop the superscript $t,t+1$ from $\A^{t,t+1}(i;j)$, and $S^{t,t+1}_{i,j}$ for short notation.

\begin{lemma}
\label{lem:worst-case}
The rate achieved when the diagonal entries of the shuffling matrix are greater than zero, i.e., when $S_{i,i}>0$ for $i\in[1:K]$,  is no larger than the worst-case rate.
\end{lemma}
\begin{proof}
The proof is straight forward, where $S_{i,i}$ is the number of data points that are needed by worker $w_i$ at times $t$ and $t+1$. Therefore, they remain in the storage of the worker $w_i$ and do not participate in the communication process. If $S_{i,i}>0$, then less number of data points are needed by worker $w_i$ and the rate is no larger than the worst-case rate, which completes the proof of the lemma.
\end{proof}
\begin{corollary}
\label{col:WC}
For the worst-case rate analysis, we can assume that  every worker is assigned only new data points, i.e., $S_{i,i} = 0$. Hence, the data conservation property in \eqref{eq:dataflow-conservation} can be written as
\begin{align}
\label{eq:data-conservation-WC}
 \sum_{j\in[1:K]\setminus j} S_{i,j}= \sum_{j\in[1:K]\setminus j} S_{j,i} = \frac{N}{K}.
\end{align}
\end{corollary}

\subsection{Closing the Gap for $m=1$}
We consider the storage value $m=1$ ($S=\frac{N}{K}$), which is the no-excess storage case considered in our previous work \cite{Shuffling/wired/master2} for any arbitrary shuffle. One can easily show that the pair $(S=\frac{N}{K},R_{\text{worst-case}}=(K-1)\frac{N}{K})$ is achievable by sending $K-1$ linear independent combinations of the $K$ data batches at time $t$, i.e., $\A^t(1),\ldots,\A^t(K)$, to satisfy any data assignment at time $t+1$. Since every worker $w_k$ has already the data batch $\A^t(k)$ already stored in its cache, it can solve for the remaining $K-1$ batches and obtain the whole data-set to store the new data assignment.

\subsection{Closing the Gap for $m=K-1$}
According to the adopted placement strategy, whenever a new data point is needed at any worker, it already has $\binom{K-2}{m-2}$ out of the total $\binom{K-1}{m-1}$ partitions, that is for the storage value $m=K-1$ ($S=(K-1)\frac{N}{K}$), only $1$ out of $K-1$ sub-points is needed.  
Furthermore, this needed data sub-point is already available at the remaining $m=K-1$ workers.
Therefore, for the $S_{i,j}$ data points assigned to worker $w_j$ and available at $w_i$, i.e., $\A(i;j)$, the data sub-batch $\A_{[1:K]\setminus \{i,j\}}(i;j)$ is the only part needed to be transmitted to $w_j$, which is available at all the workers except $w_j$.
For the worst-case scenario according to Corollary~\ref{col:WC}, we assume every worker is assigned completely new data batch, i.e., $S_{i,i}=0$ for all $i\in[1:K]$. Therefore, we can write the total part needed to be transmitted to $w_j$ as
$\cup_{i\in[1:K]\setminus j} \A_{[1:K]\setminus \{i,j\}}(i;j)$, which consists of $\frac{N}{K}$ data sub-points each of size
 $\frac{d}{K-1}$ each, and the size of $\cup_{i\in[1:K]\setminus j} \A_{[1:K]\setminus \{i,j\}}(i;j)$ (normalized by $d$) is 
\begin{align}
\label{worst-case_K-1_GC}
\vert \cup_{i\in[1:K]\setminus j} \A_{[1:K]\setminus \{i,j\}}(i;j)\vert = \frac{N}{K(K-1)}.
\end{align}
In the delivery phase, we can send the following coded data batch:
\begin{align}
\label{eq:coded-tx-K-1}
\bigoplus_{j\in[1:K]} \cup_{i\in[1:K]\setminus j} \A_{[1:K]\setminus \{i,j\}}(i;j),
\end{align}
which is useful for the $K$ workers in the same time as follows: $w_{k}$ has $\bigoplus_{j\in[1:K]\setminus k} \cup_{i\in[1:K]\setminus j} \A_{[1:K]\setminus \{i,j\}}(i;j)$ which it can subtract to recover the needed part $\cup_{i\in[1:K]\setminus k} \A_{[1:K]\setminus \{i,k\}}(i;k)$. Moreover, the size of the coded transmission in \eqref{eq:coded-tx-K-1} is the same as the size of the uncoded elements given in \eqref{worst-case_K-1_GC}  as $\frac{N}{K(K-1)}$,
which achieves the pair
$(S=(K-1)\frac{N}{K},R_{\text{worst-case}}=\frac{N}{K(K-1)})$.

\subsection{Closing the Gap for $m=K-2$}
For the storage point $m=K-2$ ($S=(K-2)\frac{N}{K}$), whenever a data point is newly assigned to a worker, it already has $\binom{K-2}{K-4}=\frac{(K-2)(K-1)}{2}$ out of $\binom{K-1}{K-3}=\frac{(K-1)(K-2)}{2}$ parts, and hence only $K-2$ parts are needed of size $\frac{2d}{(K-1)(K-2)}$ bits each.
We also assume the worst-case scenario, where according to Corollary~\ref{col:WC} every worker is assigned completely new data batch, i.e., $S_{i,i}=0$ and worker $w_i$ needs $\frac{N}{K}$ new data points for all $i\in[1:K]$.
Therefore, the total number of sub-points needed by every worker is $(K-2)\frac{N}{K}$.

\noindent $\bullet$ 
Consider the data sub-points which are considered interference to $w_k$ (neither available nor needed).
First, $w_k$ does not need nor previously assigned the data points in the batches $\A(i;j)$ where $i\neq j$ and $i,j \in[1:K]\setminus k$ (potential interference). However, not the whole data points in  $\A(i;j)$ are sent to $w_j$, since $w_j$ has already some parts of them, which are given by $\A_{\mathcal{W}}(i;j)$, where $j\in \mathcal{W}$ and $\vert\mathcal{W}\vert =K-3$.
Moreover, $w_k$ also has some parts available in its cache of $\A(i;j)$ given by $\A_{\mathcal{W}}(i;j)$, where $k\in \mathcal{W}$ (do not cause interference).
As a summary, the part of $\A(i;j)$, where $i\neq j$ and $i,j \in[1:K]\setminus k$, that is considered interference to $w_k$ is given by $\A_{[1:K]\setminus\{i,j,k\}}(i;j)$, and hence the total interference faced by $w_k$ is
\begin{align}
\mathcal{I}(k)=\underset{\substack{i,j\in[1:K]\setminus k\\ i\neq j}}{\cup} \A_{[1:K]\setminus\{i,j,k\}}(i;j).
\end{align}

\noindent $\bullet$ Next, we organize these interference sub-batches according to the workers that need them as in Figure~\ref{fig:interference-align}a. Worker $w_j$, where $j\in[1:K]\setminus k$, needs the following sub-batches causing interference to $w_k$:
\begin{align}
\mathcal{I}(j;k)=\underset{i\in[1:K]\setminus \{k,j\}}{\cup} \A_{[1:K]\setminus\{i,j,k\}}(i;j),
 \end{align}
  which consists of data sub-points of size $\frac{2d}{(K-1)(K-2)}$ each and total number given by
\begin{align}
\label{eq:I(j;k)}
I_{j;k}=
\sum_{i\in[1:K]\setminus\{k,j\}}S_{i,j} = N/K - S_{k,j}= \sum_{i\in[1:K]\setminus\{k,j\}}S_{k,i}.
\end{align}
Note that  $\mathcal{I}(j;k)$ serves as: a) interference to $w_k$, b) useful for $w_j$; and c) available at all the remaining workers.
Also, the total interference faced by $w_k$ can be written as $\mathcal{I}(k)= \cup_{j\in[1:K]\setminus k}\mathcal{I}(j;k)$ which consists of data sub-points of size $\frac{2d}{(K-1)(K-2)}$ each and total number given by
\begin{align}
I_k = \sum_{j\in[1:K]\setminus k}I_{j;k} = \sum_{j\in[1:K]\setminus k} (N/K - S_{k,j}) = (K-2)\frac{N}{K}.
\end{align}

\begin{figure}[t]
  \begin{center}
  \includegraphics[width=0.8\columnwidth]{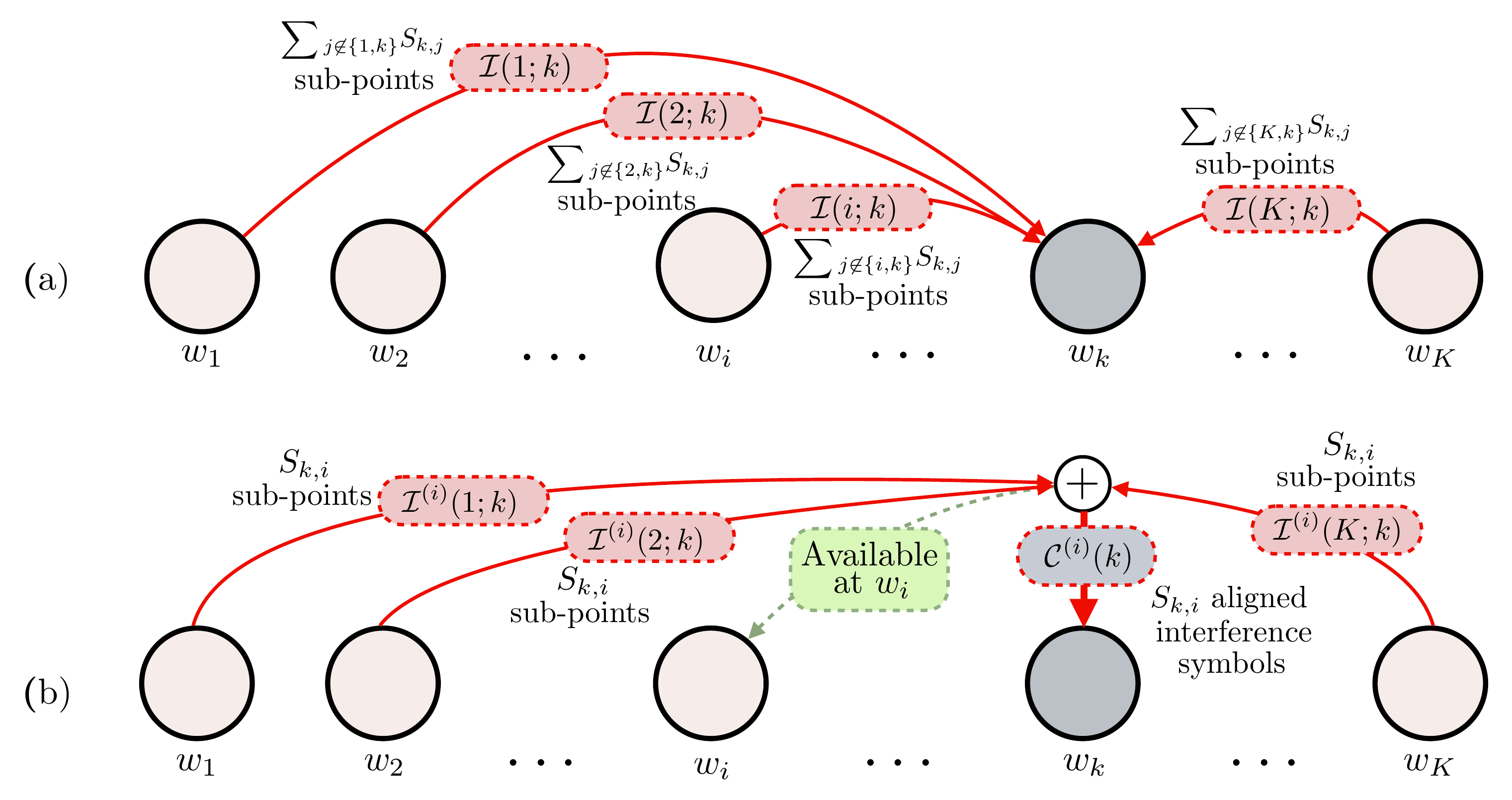}
\caption{In $(a)$, the  interference sub-batches seen by $w_k$ are organized according to the workers that need them, such that $\mathcal{I}(i;k)$ for $i\neq k$ is the data needed to be delivered to $w_i$ while causing interference to $w_k$. In $(b)$, we pick $S_{k,i}$ sub-points from each $\mathcal{I}(j;k)$ (labeled as $\mathcal{I}^{(i)}(j;k)$) where $j\not\in\{i,k\}$, and align them into $S_{k,i}$ coded symbols labeled by the set $\mathcal{C}^{(i)}(k)$, which serve as aligned interference for $w_k$, available at $w_i$, and useful for the remaining workers.
\label{fig:interference-align}}
  \end{center}
\end{figure}

\noindent $\bullet$ Following Example~\ref{ex3}, we apply a similar interference alignment argument.  We first break $\mathcal{I}(j;k)$ for every $j\in[1:K]\setminus k$ into $K-2$ partitions labeled as $\mathcal{I}^{(i)}(j;k)$ for $i\in[1:K]\setminus \{j,k\}$. The number of sub-points in $\mathcal{I}^{(i)}(j;k)$ is $S_{k,i}$ which satisfies the total size of $\mathcal{I}(j;k)$ given in \eqref{eq:I(j;k)}. As shown in Figure~\ref{fig:interference-align}b, we generate $S_{k,i}$ coded sub-points for every $i\in[1:K]\setminus k$ as follows:
\begin{align}
S_{k,i} \text{ coded sub-points}: \quad \mathcal{C}^{(i)}(k) = \bigoplus_{j\in [1:K]\setminus \{k,i\}}\mathcal{I}^{(i)}(j;k), \quad \forall i\in[1:K]\setminus k.
\end{align}
Note that $\mathcal{C}^{(i)}(k)$ is a coded sub-batch serves as: a) aligned interference to $w_k$, b) available at $w_i$ as $j\neq i$ in the above summation; and c) useful for all the remaining workers as follows: worker $w_{\ell}$ for $\ell \not\in \{i, k\}$ has $\bigoplus_{j\in [1:K]\setminus \{k,i,\ell\}}\mathcal{I}^{(i)}(j;k)$ so it can subtract from $\mathcal{C}^{(i)}(k)$ to get the needed part $\mathcal{I}^{(i)}(\ell;k)$.

\noindent $\bullet$ The total size of $\cup_{i\in[1:K]\setminus k}\mathcal{C}^{(i)}(k)$  is  $\sum_{i\in[1:K]\setminus k} S_{k,i} =\frac{N}{K}$ coded sub-points, which aligns the $I_k=(K-2)\frac{N}{K}$ total interference sub-points seen by $w_k$, i.e., $\mathcal{I}(k)$ into $\frac{N}{K}$ coded sub-points.
In the same time, these $\frac{N}{K}$ coded sub-points serve, for each remaining worker $w_j$ for $j\neq k$, as $\sum_{i\in[1:K]\setminus \{j,k\}}S_{k,i}= \frac{N}{K} - S_{k,j}$ useful sub-points given by $\cup_{i\in[1:K]\setminus \{k,j\}}\mathcal{C}^{(i)}(k)$, while the remaining $S_{k,j}$ sub-points, given by $\mathcal{C}^{(j)}(k)$, are available at $w_j$'s cache.

\noindent $\bullet$ By aligning all the interference seen by all the workers, i.e., generating the coded batches $\cup_{i\in[1:K]\setminus k}\mathcal{C}^{(i)}(k)$ for all $k\in[1:K]$, we get a total number of $N$ coded sub-points seen as follows by every worker $w_j$: a) $\frac{N}{K}$ aligned interference coded sub-points, b) $\sum_{k\in[1:K]\setminus j} S_{k,j} =\frac{N}{K}$ available sub-points; and c) $\sum_{k\in[1:K]\setminus j}\left(\frac{N}{K} - S_{k,j}\right) = (K-2)\frac{N}{K}$ useful sub-points, which satisfies the total number of sub-points needed in the worst case as discussed in the beginning. Since out of all the $N$ coded sub-points every worker already has $\frac{N}{K}$ of them, then the $N$ coded sub-points can be sent in only $(K-1)\frac{N}{K}$ linear independent combinations of size $\frac{2d}{(K-1)(K-2)}$ each, where the interference sub-points occupy $\frac{N}{K}$ dimensions, while the useful sub-points occupy $(K-2)\frac{N}{K}$ dimensions. As a result, the total rate achieved is $\frac{2N}{K(K-2)}d$ bits, which achieves the pair
$(S=(K-2)\frac{N}{K},R_{\text{worst-case}}=\frac{2N}{K(K-2)})$.

Now that we have closed the gap between the bounds in Theorems~\ref{thm1} and \ref{thm2} for $S=m\frac{N}{K}$, where $m\in\{1,K-2,K-1\}$, which covers all the storage values for $K<5$, while for $K\geq 5$ we can do the same analysis as in Section~\ref{sec:gap-analysis} to obtain the gap ratio similar to \eqref{eq:gap-ratio} as follows:
\begin{align}
\label{eq:gap-ratio2}
\frac{R_{\text{worst-case}}^{\text{upper}}}{R_{\text{worst-case}}^{\text{lower}}}=1 + \frac{1}{K-j}\left(1 -\frac{(K+1)(j-1)}{(K-1)(j+1)}\right) =
1+ \frac{2}{(K-1)(j+1)}, \quad j\in[2:K-1],
\end{align}
which is maximized for $j=2$ to obtain the maximum gap ratio as $1+ \frac{2}{(K-1)(3)}= \frac{K-\frac{1}{3}}{K-1}$ for $K\geq 5$ which completes the proof of Theorem~\ref{thm4}.

\end{appendices}

\end{document}